\documentclass[sigplan,9pt]{acmart}

\usepackage[english]{babel}
\usepackage[T1]{fontenc}
\usepackage[utf8]{inputenc}
\usepackage[l2tabu, orthodox]{nag}
\usepackage[ruled,noline]{algorithm2e}

\usepackage{bbm}
\usepackage{proof}
\usepackage{csquotes}
\usepackage{mathtools}
\usepackage{booktabs}
\usepackage{subcaption}
\usepackage{multirow}

\usepackage{tikz}
\usetikzlibrary{automata,positioning,shapes,calc,fit,arrows.meta,backgrounds}
\usepackage{pgfplots}
\usepgfplotslibrary{fillbetween}
\pgfplotsset{compat=1.13}

\tikzstyle{state}+=[minimum size=20pt,inner sep=2pt]
\tikzstyle{action}=[font=\small,inner sep=0pt,outer sep=3pt]
\tikzstyle{actionnode}=[circle,draw=black,fill=black,minimum size=1mm,inner sep=0,outer sep=0]
\tikzstyle{actionedge}=[-,draw]
\tikzstyle{prob}=[font=\scriptsize,inner sep=0pt,outer sep=1pt]
\tikzstyle{probedge}=[draw]
\tikzstyle{directedge}=[draw]
\tikzset{chainarrow/.tip={Stealth[length=3pt]}}
\tikzset{>=chainarrow}

\tikzstyle{discontinuity-limit}=[circle,inner sep=0pt,minimum size=3pt,fill=white,draw=black]
\tikzstyle{discontinuity-value}=[circle,inner sep=0pt,minimum size=3pt,fill=black]

\usepackage[textsize=tiny,final]{todonotes}
\presetkeys{todonotes}{noline}{}
\input{macros}

\begin{document}

%% Title information
%\title[Conditional Value-at-Risk]{Optimizing the Conditional Value-at-Risk}
\title[CVaR for Reachability and Mean Payoff in MDP]{Conditional Value-at-Risk for Reachability and Mean Payoff in Markov Decision Processes}
%\subtitle{Subtitle}
%\subtitlenote{with subtitle note}

%% Author information
\author{Jan K{\v{r}}et\'insk\'y}
\orcid{0000-0002-8122-2881}
\affiliation{
  \position{}
  \department{Institut für Informatik (I7)}
  \institution{Technische Universität München}
  \streetaddress{Boltzmannstr. 3}
  \city{Garching bei München}
  \state{Bavaria}
  \postcode{85748}
  \country{Germany}
}
\email{jan.kretinsky@in.tum.de}

\author{Tobias Meggendorfer}
\orcid{0000-0002-1712-2165}
\affiliation{
  \position{}
  \department{Institut für Informatik (I7)}
  \institution{Technische Universität München}
  \streetaddress{Boltzmannstr. 3}
  \city{Garching bei München}
  \state{Bavaria}
  \postcode{85748}
  \country{Germany}
}
\email{tobias.meggendorfer@in.tum.de}

\begin{abstract}
We present the \emph{conditional value-at-risk (CVaR)} in the context of Markov chains and Markov decision processes with reachability and mean-payoff objectives.
CVaR quantifies risk by means of the expectation of the worst $p$-quantile.
As such it can be used to design risk-averse systems.
We consider not only CVaR constraints, but also introduce their conjunction with expectation constraints and quantile constraints (value-at-risk, VaR).
We derive lower and upper bounds on the computational complexity of the respective decision problems and characterize the structure of the strategies in terms of memory and randomization.
\end{abstract}

%% 2012 ACM Computing Classification System (CSS) concepts
%% Generate at 'http://dl.acm.org/ccs/ccs.cfm'.



\copyrightyear{2018}
\acmYear{2018}
\setcopyright{none}
\acmConference[LICS '18]{33rd Annual ACM/IEEE Symposium on Logic in Computer Science}{July 9--12, 2018}{Oxford, United Kingdom}
\acmBooktitle{LICS '18: 33rd Annual ACM/IEEE Symposium on Logic in Computer Science, July 9--12, 2018, Oxford, United Kingdom}
\acmPrice{15.00}
\acmDOI{10.1145/3209108.3209176}
\acmISBN{978-1-4503-5583-4/18/07}

\maketitle

\section{Introduction}

\paragraph{Markov decision processes (MDP)} are a standard formalism for modelling stochastic systems featuring non-determinism.
The fundamental problem is to design a strategy resolving the non-determi\-nistic choices so that the systems' behaviour is optimized with respect to a given objective function, or, in the case of multi-objective optimization, to obtain the desired trade-off.
The objective function (in the optimization phrasing) or the query (in the decision-problem phrasing) consists of two parts. 
First, a payoff is a measurable function assigning an outcome to each run of the system.
It can be real-valued, such as the \emph{long-run average reward} (also called \emph{mean payoff}), or a two-valued predicate, such as \emph{reachability}.
Second, the payoffs for single runs are combined into an overall outcome of the strategy, typically in terms of \emph{expectation}.
The resulting objective function is then for instance the expected long-run average reward, or the probability to reach a given target state.

\paragraph{Risk-averse control} aims to overcome one of the main disadvantages of the expectation operator, namely its ignorance towards the incurred risks, intuitively phrased as a question \emph{\enquote{How bad are the bad cases?}}
While the standard deviation (or variance) quantifies the spread of the distribution, it does not focus on the bad cases and thus fails to capture the risk.
There are a number of quantities used to deal with this issue:
\begin{itemize}
	\item The \emph{worst-case} analysis (in the financial context known as discounted maximum loss) looks at the payoff of the worst possible run.
	While this makes sense in a fully non-deterministic environment and lies at the heart of verification, in the probabilistic setting it is typically unreasonably pessimistic, taking into account events happening with probability $0$, e.g., never tossing head on a fair coin.
	\item The \emph{value-at-risk} ($\VaR$) denotes the worst $p$-quantile for some $p \in [0,1]$.
	For instance, the value at the $0.5$-quantile is the median, the $0.05$-quantile (the \emph{vigintile} or \emph{ventile}) is the value of the best run among the $5\%$ worst ones.
	As such it captures the \enquote{reasonably possible} worst-case.
	See Fig.~\ref{fig:cvar_explain} for an example of VaR for two given probability density functions.
	There has been an extensive effort spent recently on the analysis of MDP with respect to VaR and the re-formulated notions of quantiles, percentiles, thresholds, satisfaction view etc., see below.
	Although VaR is more realistic, it tends to ignore outliers too much, as seen in Fig.~\ref{fig:cvar_explain} on the right.
	VaR has been characterized as \emph{\enquote{seductive, but dangerous}} and \emph{\enquote{not sufficient to control risk}} \cite{seductive}.
	\item The \emph{conditional value-at-risk} (average value-at-risk, expected shortfall, expected tail loss) answers the question \emph{\enquote{What to expect in the bad cases?}}
	It is defined as the expectation over all events worse than the value-at-risk, see Fig.~\ref{fig:cvar_explain}.
	As such it describes the lossy tail, taking outliers into account, weighted respectively.
	In the degenerate cases, CVaR for $p = 1$ is the expectation and for $p = 0$ the (probabilistic) worst case.
	It is an established risk metric in finance, optimization and operations research, e.g. \cite{MAFI:MAFI068,Rockafellar00optimizationof}, and \emph{\enquote{is considered to be a more consistent measure of risk}} \cite{Rockafellar00optimizationof}.
	Recently, it started permeating to areas closer to verification, e.g. robotics \cite{DBLP:conf/icra/CarpinCP16}.
\end{itemize}

\begin{figure}
	\begin{subfigure}[!t]{0.49\columnwidth}
		\begin{tikzpicture}
			\begin{axis}[axis x line=middle, axis y line=middle,
					axis line style={->},
					x label style={at={(axis description cs:0,0)},anchor=north west},
					y label style={at={(axis description cs:0,0)},rotate=90,anchor=south west},
					ymajorticks=false,
					height=2.8cm,width={1.2\textwidth},
					xmin=0,xmax=2.2,xtick={0,1,2},xlabel={value},
					ymin=0,ymax=0.5,ytick style={draw=none},ylabel={density},
					samples=50]
				\addplot[domain=0:0.55, smooth, fill=lightgray, draw=none]{(e^(1 - (x - 1)^2) - 1)^4 / (20)} \closedcycle;
				\addplot[domain=0:2, smooth]{(e^(1 - (x - 1)^2) - 1)^4 / (20)};

				\draw [dashed] (axis cs:0.45,0) -- (axis cs:0.45,0.32) node[rotate=90,anchor=south] {$\CVaR$} -- (axis cs:0.45,0.47);
				\draw [dashed] (axis cs:0.55,0) -- (axis cs:0.55,0.38) node[rotate=90,anchor=north] {$\VaR$} -- (axis cs:0.55,0.47);

				\draw (axis cs:0,0) -- (axis cs:1,0);
			\end{axis}
		\end{tikzpicture}
	\end{subfigure} %
	\begin{subfigure}[!t]{0.49\columnwidth}
		\begin{tikzpicture}
			\begin{axis}[axis x line=middle, axis y line=middle,
					axis line style={->},
					x label style={at={(axis description cs:0,0)},anchor=north west},
					y label style={at={(axis description cs:0,0)},rotate=90,anchor=south west},
					ymajorticks=false,
					height=2.8cm,width={1.2\textwidth},
					xmin=0,xmax=2.2,xtick={0,1,2},xlabel={value},
					ymin=0,ymax=0.5,ytick style={draw=none},ylabel={density},
					samples=50,]
				\addplot[domain=0:0.25, smooth, fill=lightgray, draw]{sin(x*720)^2 / 4};

				\addplot[domain=1:1.1, smooth, fill=lightgray, draw=none]{(e^(-(x - 1.5)^2) - e^(-0.25))^2 * 7.5} \closedcycle;
				\addplot[domain=1:2, smooth]{(e^(-(x - 1.5)^2) - e^(-0.25))^2 * 7.5};

				\draw [dashed] (axis cs:0.3,0) -- (axis cs:0.3,0.32) node[rotate=90,anchor=north] {$\CVaR$} -- (axis cs:0.3,0.47);
				\draw [dashed] (axis cs:1.1,0) -- (axis cs:1.1,0.38) node[rotate=90,anchor=south] {$\VaR$} -- (axis cs:1.1,0.47);

				\draw (axis cs:0,0) -- (axis cs:2,0);
			\end{axis}
		\end{tikzpicture}
	\end{subfigure}
	\caption{Illustration of VaR and CVaR for some random variables.}
	\label{fig:cvar_explain}
\end{figure}
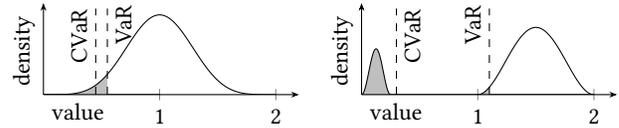

\paragraph{Our contribution}

In this paper, we investigate optimization of MDP with respect to CVaR as well as the respective trade-offs with expectation and VaR.
We study the VaR and CVaR operators for the first time with the payoff functions of weighted reachability and mean payoff, which are fundamental in verification.
Moreover, we cover both the single-dimensional and the multi-dimensional case.

Particularly, we define CVaR for MDP and show the peculiarities of the concept.
Then we study the computational complexity and the strategy complexity for various settings, proving the following:
\begin{itemize}
	\item
	The single dimensional case can be solved in polynomial time through linear programming, see Section~\ref{sec:single}.
	\item
	The multi-dimensional case is NP-hard, even for CVaR-only constraints.
	Weighted reachability is NP-complete and we give PSPACE and EXPSPACE upper bounds for mean payoff with CVaR and expectation constraints, and with additional VaR constraints, respectively, see Section~\ref{sec:multi}.
	(Note that already for the sole VaR constraints only an exponential algorithm is known; the complexity is an open question and not even NP-hardness is known~\cite{DBLP:journals/fmsd/RandourRS17,DBLP:conf/lics/ChatterjeeKK15}.)
	\item We characterize the strategy requirements, both in terms of memory, ranging from memoryless, over constant-size to infinite memory, and the required degree of randomization, ranging from fully deterministic strategies to randomizing strategies with stochastic memory update.
\end{itemize}

While dealing with the CVaR operator, we encountered surprising behaviour, preventing us to trivially adapt the solutions to the expectation and VaR problems:
\begin{itemize}
	\item Compared to, e.g., expectation and VaR, CVaR does not behave linearly w.r.t. stochastic combination of strategies.
	%\item Quite complex strategies are already needed in the single-dimensional case (memory and randomization).
	\item A conjunction of CVaR constraints already is NP-hard, since it can force a strategy to play deterministically.
\end{itemize}

%We present tight bounds on the computational complexity and structure of the strategies for the one-dimensional case with various objectives and give an outlook for the multi-dimensional case.
%We consider single dimensional queries in Sec.~\ref{sec:single}.
%There, we show PTIME solutions for all types of queries and give precise bounds on the class of optimal strategies.
%Afterwards, we deal with the harder multidimensional case in Sec.~\ref{sec:multi}.
%We prove NP-hardness for all queries as well as NP-completeness for weighted reachability and provide further insights in the mean payoff setting. \todo{Are we exactly doing this?}

%=====
%%?
%In verification , the typical objective are \emph{reachability}, to which e.g. linear temporal logic reduces to, and 
%In the case of reachability 
%
%recent interest
%=====================

\subsection{Related work}

\paragraph{Worst case}
Risk-averse approaches optimizing the worst case together with expectation have been considered in beyond-worst-case and beyond-almost-sure analysis investigated in both the single-dimensional~\cite{DBLP:journals/iandc/BruyereFRR17}
%[V. Bruye`re, E. Filiot, M. Randour, and J. Raskin. Meet your expectations with guarantees: Beyond worst-case synthesis in quantitative games.] 
and in the multi-dimensional~\cite{DBLP:conf/lics/ClementeR15} setup.
%[L. Clemente and J.-F. Raskin. Multidimensional beyond worst-case and almost-sure problems for mean-payoff objectives. In LICS. To appear, 2015.]

%\todo{It may still be appropriate in the non-deterministic setup where some knowledge of the probability distribution on the non-deterministic choices is known, as in the beyond-worst-case synthesis \cite{raskin}.}

\paragraph{Quantiles}
The decision problem related to VaR has been phrased in probabilistic verification mostly in the form \emph{\enquote{Is the probability that the payoff is higher than a given value threshold more than a given probability threshold?}}
The total reward gained attention both in the verification community~\cite{DBLP:conf/fossacs/UmmelsB13,DBLP:conf/icalp/HaaseK15,DBLP:conf/tacas/Baier0KW17} and recently in the AI community~\cite{DBLP:conf/aaai/GilbertWX17,DBLP:journals/corr/abs-1711-05788}.
%
%This question gained recent attention in the verification community in various settings.
%As for the total reward in MDP, \cite{DBLP:conf/fossacs/UmmelsB13} introduces the question and presents a polynomial solution in the qualitative case and an exponential algorithm for the quantitative case.
%This problem is also considered recently in AI community \cite{DBLP:conf/aaai/GilbertWX17,DBLP:journals/corr/abs-1711-05788}.
%The exact computational complexity is discussed in \cite{DBLP:conf/icalp/HaaseK15}.
%A related problem arises when absorption is not ensured almost surely, then the total reward conditioned on absorption is of interest \cite{DBLP:conf/tacas/Baier0KW17} 
%
Multi-dimensional percentile queries are considered for various objectives, such as mean-payoff, limsup, liminf, shortest path in~\cite{DBLP:journals/fmsd/RandourRS17}; for the specifics of two-dimensional case and their interplay, see~\cite{DBLP:conf/nfm/BaierDDKK14}.
Quantile queries for more complex constraints have also been considered, namely their conjunctions~\cite{FKR95,DBLP:journals/corr/abs-1104-3489}, conjunctions with expectations~\cite{DBLP:conf/lics/ChatterjeeKK15} or generally Boolean expressions~\cite{DBLP:conf/lics/HaaseKL17}.
Some of these approaches have already been practically applied and found useful by domain experts~\cite{DBLP:conf/fase/BaierDKDKMW14,DBLP:conf/csl/BaierDK14}.

\paragraph{CVaR}

There is a body of work that optimizes CVaR in MDP.
However, to the best of our knowledge, all the approaches (1) focus on the single-dimensional case, (2) disregard the expectation, and (3) treat neither reachability nor mean payoff.
They focus on the discounted~\cite{DBLP:journals/mmor/BauerleO11}, total~\cite{DBLP:conf/icra/CarpinCP16}, or immediate~\cite{DBLP:journals/ajcm/KageyamaFKT11} reward, as well as extend the results to continuous-time models~\cite{DBLP:journals/siamjo/HuangG16,DBLP:journals/siamco/MillerY17}.
This work comes from the area of optimization and operations research, with the notable exception of~\cite{DBLP:conf/icra/CarpinCP16}, which focuses on the total reward. 
Since the total reward generalizes weighted reachability, \cite{DBLP:conf/icra/CarpinCP16} is related to our work the most.
However, it provides only an approximation solution for the one-dimensional case, neglecting expectation and the respective trade-offs.

Further, CVaR is a topic of high interest in finance, e.g., \cite{Rockafellar00optimizationof,seductive}.
The central difference is that there variations of portfolios (i.e. the objective functions) are considered while leaving the underlying random process (the market) unchanged.
This is dual to our problem, since we fix the objective function and now search for an optimal random process (or the respective strategy).
%Consequently, although CVaR is convex in the former setting, it is not convex in our setting convex, i.e., with respect to the random process.

\paragraph{Multi-objective expectation}

In the last decade, MDP have been extensively studied generally in the setting of multiple objectives, which provides some of the necessary tools for our trade-off analysis.
Multiple objectives have been considered for both qualitative payoffs, such as reachability and LTL~\cite{DBLP:journals/lmcs/EtessamiKVY08}, as well as quantitative payoffs, such as mean payoff~\cite{DBLP:journals/corr/abs-1104-3489}, discounted sum~\cite{DBLP:conf/lpar/ChatterjeeFW13}, or total reward~\cite{DBLP:conf/tacas/ForejtKNPQ11}.
Variance has been introduced to the landscape in~ \cite{DBLP:conf/lics/BrazdilCFK13}.

%!TEX root = main.tex

\section{Preliminaries}

Due to space constraints, some proofs and explanations are shortened or omitted when clear and can be found in the appendix.

\subsection{Basic definitions}

We mostly follow the definitions of \cite{DBLP:conf/lics/ChatterjeeKK15,DBLP:journals/corr/abs-1104-3489}.
$\Naturals, \Rationals, \Reals$ are used to denote the sets of positive integers, rational and real numbers, respectively.
For $n \in \Naturals$, let $[n]=\{1, \ldots, n\}$.
Further, $k_j$ refers to $k \cdot e_j$, where $e_j$ is the unit vector in dimension $j$.

We assume familiarity with basic notions of probability theory, e.g., \emph{probability space} $(\ProbSpace, \Salgebra, \measure)$, \emph{random variable} $F$, or \emph{expected value} $\Expectation$.
The set of all distributions over a countable set $C$ is denoted by $\Distributions(C)$.
Further, $d \in \Distributions(C)$ is Dirac if $d(c) = 1$ for some $c \in C$.
To ease notation, for functions yielding a distribution over some set $C$, we may write $f(\cdot, c)$ instead of $f(\cdot)(c)$ for $c \in C$.

\paragraph{Markov chains}
A \emph{Markov chain} (MC) is a tuple $\MC = (\States, \trans, \mu_0)$, where
	$\States$ is a countable set of states\footnote{We allow the state set to be countable for the formal definition of strategies on MDP.
		When dealing with Markov Chains in queries, we only consider finite state sets.},
	$\trans : \States \to \Distributions(\States)$ is a probabilistic transition function, and
	$\mu_0 \in \Distributions(\States)$ is the initial probability distribution.
The SCCs and BSCCs of a MC are denoted by $\SCCs$ and $\BSCCs$, respectively~\cite{Puterman-book}.

A \emph{run} in $\MC$ is an infinite sequence $\run = s_1 s_2 \cdots$ of states, we write $\run_i$ to refer to the $i$-th state $s_i$.
A \emph{path} $\finitePath$ in $\MC$ is a finite prefix of a run $\run$.
Each path $\finitePath$ in $\MC$ determines the set $\Cone(\finitePath)$ consisting of all runs that start with $\finitePath$.
To $\MC$, we associate the usual probability space $(\Runs, \Salgebra, \Prob)$, where $\Runs$ is the set of all runs in $\MC$, $\Salgebra$ is the $\sigma$-field generated by all $\Cone(\finitePath)$, and $\Prob$ is the unique probability measure such that $\Prob(\Cone(s_1 \cdots s_k)) = \mu_0(s_1) \cdot \prod_{i=1}^{k-1} \trans(s_i, s_{i+1})$.
Furthermore, $\LtlEventually B$ ($\LtlEventually \LtlAlways B$) denotes the set of runs which eventually reach (eventually remain in) the set $B \subseteq \States$, i.e.\ all runs where $\run_i \in B$ for some $i$ (there exists an $i_0$ such that $\run_i \in B$ for all $i \geq i_0$).

\paragraph{Markov decision processes}
A \emph{Markov decision process} (MDP) is a tuple $\MDP = (\States, \Actions, \AvAct, \Trans, \initstate)$ where
	$\States$ is a finite set of states,
	$\Actions$ is a finite set of actions,
	$\AvAct : \States \rightarrow 2^\Actions \setminus \{\emptyset\}$ assigns to each state $s$ the set $\AvAct(s)$ of actions enabled in $s$ so that $\{\AvAct(s) \mid s \in \States\}$ is a partitioning of $\Actions$\footnote{In other words, each action is associated with exactly one state.},
	$\Trans : \Actions \rightarrow \Distributions(\States)$ is a probabilistic transition function that given an action $a$ yields a probability distribution over the successor states, and
	$\initstate$ is the initial state of the system.

A \emph{run} $\run$ of $\MDP$ is an infinite alternating sequence of states and actions $\run = s_1 a_1 s_2 a_2 \cdots$ such that for all $i \geq 1$, we have $a_i \in \AvAct(s_i)$ and $\Delta(a_i, s_{i+1}) > 0$.
Again, $\run_i$ refers to the $i$-th state visited by this particular run.
A \emph{path} of length $k$ in $\MDP$ is a finite prefix $\finitePath = s_1 a_1 \cdots a_{k-1} s_k$ of a run in $G$.

\paragraph{Strategies and plays.}
Intuitively, a strategy in an MDP $\MDP$ is a \enquote{recipe} to choose actions based on the observed events.
Usually, a strategy is defined as a function $\strategy : (\States \Actions)^*\States \to \Distributions(\Actions)$ that given a finite path $\finitePath$, representing the history of a play, gives a probability distribution over the actions enabled in the last state.
We adopt the slightly different, though equivalent~\cite[Sec.~6]{DBLP:journals/corr/abs-1104-3489} definition from~\cite{DBLP:conf/lics/ChatterjeeKK15}, which is more convenient for our setting.

Let $\Memory$ be a countable set of \emph{memory elements}.
A \emph{strategy} is a triple $\strategy = (\strategy_u, \strategy_n, \alpha)$, where
	$\strategy_u: \Actions \times \States \times \Memory \to \Distributions(\Memory)$ and
	$\strategy_n: \States \times \Memory \to \Distributions(\Actions)$ are \emph{memory update} and \emph{next move} functions, respectively, and
	$\alpha \in \Distributions(\Memory)$ is the initial memory distribution.
We require that, for all $(s, m) \in \States \times \Memory$, the distribution $\strategy_n(s,m)$ assigns positive values only to actions available at $s$, i.e.\ $\supp \strategy_n(s,m) \subseteq \AvAct(s)$.

A \emph{play} of $\MDP$ determined by a strategy $\strategy$ is a Markov chain $\MDP^\strategy = (\States^\strategy, \trans^\strategy, \mu_0^\strategy)$, where
	the set of states is $\States^\strategy = \States \times \Memory \times \Actions$,
	the initial distribution $\mu_0$ is zero except for $\mu_0^\strategy(\initstate, m, a) = \alpha(m) \cdot \strategy_n(\initstate, m, a)$, and
	the transition probability from $s^\strategy = (s, m, a)$ to ${s'}^\strategy = (s', m', a')$ is $\trans^\strategy(s^\strategy, {s'}^\strategy) = \Trans(a, s') \cdot \strategy_u(a, s', m, m') \cdot  \strategy_n(s',m', a')$.
Hence, $\MDP^\strategy$ starts in a location chosen randomly according to $\alpha$ and $\strategy_n$.
In state $(s, m, a)$ the next action to be performed is $a$, hence the probability of entering $s'$ is $\Trans(a, s')$.
The probability of updating the memory to $m'$ is $\strategy_u(a, s', m, m')$, and the probability of selecting $a'$ as the next action is $\strategy_n(s',m', a')$. Since these choices are independent, and thus we obtain the product above.

Technically, $\MDP^\strategy$ induces a probability measure $\Prob^\strategy$ on $\States^\strategy$.
Since we mostly work with the corresponding runs in the original MDP, we overload $\Prob^\strategy$ to also refer to the probability measure obtained by projecting onto $\States$.
Further, \enquote{almost surely} etc. refers to happening with probability 1 according to $\Prob^\strategy$.
The expected value of a random variable $X : \Runs \to \Reals$ is $\Expectation^\strategy[X] = \int_\Runs X \ d\Prob^\strategy$.

A convex combinations of two strategies $\strategy_1$ and $\strategy_2$, written as $\strategy_\lambda = \lambda \strategy_1 + (1-\lambda) \strategy_2$, can be obtained by defining the memory as $\Memory_\lambda = \{1\} \times \Memory_1 \union \{2\} \times \Memory_2$, randomly choosing one of the two strategies via the initial memory distribution $\alpha_\lambda$ and then following the chosen strategy.
Clearly, we have that $\Prob^{\strategy_\lambda} = \lambda \Prob^{\strategy_1} + (1-\lambda) \Prob^{\strategy_2}$.

\paragraph{Strategy types.}
A strategy $\strategy$ may use infinite memory $\Memory$, and both $\strategy_u$ and $\strategy_n$ may randomize.
The strategy $\strategy$ is
\begin{itemize}
	\item \emph{deterministic-update}, if $\alpha$ is Dirac and the memory update function $\strategy_u$ gives a Dirac distribution for every argument;
	\item \emph{deterministic}, if it is deterministic-update and the next move function $\strategy_n$ gives a Dirac distribution for every argument.
\end{itemize}
A \emph{stochastic-update} strategy is a strategy that is not necessarily deterministic-update and \emph{randomized} strategy is a strategy that is not necessarily deterministic.
We also classify the strategies according to the size of memory they use.
Important subclasses are
	\emph{memoryless} strategies, in which $\Memory$ is a singleton,
	\emph{$n$-memory} strategies, in which $\Memory$ has exactly $n$~elements, and
	\emph{finite-memory} strategies, in which $\Memory$ is finite.

\paragraph{End components.}
A tuple $(T, B)$ where $\emptyset \neq T \subseteq S$ and $\emptyset \neq B \subseteq \bigcup_{t\in T} \AvAct(t)$ is an \emph{end component} of the MDP $\MDP$ if
	(i)~for all actions $a \in B$, $\Trans(a, s') > 0$ implies $s'\in T$; and
	(ii)~for all states $s, t \in T$ there is a path $\finitePath = s_1 a_1\cdots a_{k-1} s_k \in (TB)^{k-1}T$ with $s_1 = s$, $s_k = t$.
An end component $(T, B)$ is a \emph{maximal end component (MEC)} if $T$ and $B$ are maximal with respect to subset ordering.
Given an MDP, the set of MECs is denoted by $\MECs$.
By abuse of notation, $s \in M$ refers to all states of a MEC $M$, while $a \in M$ refers to the actions.

\begin{remark} \label{rem:mec_decomposition_ptime}
	Computing the maximal end component (MEC) decomposition of an MDP, i.e.\ the computation of $\MECs$, is in P~\cite{DBLP:journals/jacm/CourcoubetisY95}.
\end{remark}

%\begin{remark} \label{rem:mec_quotient}
%	The \emph{MEC quotient} of some MDP $\MDP$ refers to another MDP $\widehat\MDP$ in which each MEC has been replaced by a single state, representing the whole MEC, and transitions adapted accordingly~\cite{dA97a}.
%	This construction preserves many properties, e.g., reachability~\cite{DBLP:conf/cav/AshokCDKM17,DBLP:conf/atva/KretinskyM17}.
%	See \cite{APPENDIX} for a formal definition and relevant properties.
%\end{remark}

\begin{remark} \label{rem:mec_almost_sure}
	For any MDP $\MDP$ and strategy $\strategy$, a run almost surely eventually stays in one MEC, i.e.\ $\Prob^\strategy[\Union_{M_i \in \MECs} \LtlEventually\LtlAlways M_i] = 1$~\cite{Puterman-book}.
\end{remark}

\subsection{Random variables on Runs}

We introduce two standard random variables, assigning a value to each run of a Markov Chain or Markov Decision Process.

\paragraph{Weighted reachability.}
Let $\reachSet \subseteq \States$ be a set of target states and $\reward : \reachSet \mapsto \Rationals$ be a reward function.
Define the random variable $\ObjFun^\QueryObjReach$ as $\ObjFun^\QueryObjReach(\run) = \reward(\min_i \{\run_i \mid \run_i \in \reachSet\})$, if such an $i$ exists, and $0$ otherwise.
Informally, $\ObjFun^\QueryObjReach$ assigns to each run the value of the first visited target state, or $0$ if none.
$\ObjFun^\QueryObjReach$ is measurable and discrete, as $S$ is finite~\cite{Puterman-book}.
Whenever we are dealing with weighted reachability, we assume w.l.o.g. that all target states are absorbing, i.e.\ for any $s \in \reachSet$ we have $\trans(s, s) = 1$ for MC and $\Trans(a, s) = 1$ for all $a \in \AvAct(s)$ for MDP.

%Note that apart from some technicalities this definition is equal to defining $\reachSet = \{s \in \States \mid \reward(s) \neq 0\}$.

\paragraph{Mean payoff} (also known as \emph{long-run average reward}, and \emph{limit average reward}).
Again, let $\reward : \States \mapsto \Rationals$ be a reward function.
The mean payoff of a run $\run$ is the average reward obtained per step, i.e. $\ObjFun^\QueryObjMean(\run) = \liminf_{n \to \infty} \frac{1}{n} \sum_{i = 1}^n \reward(\run_i).$
The $\liminf$ is necessary, since $\lim$ may not be defined in general.
Further, $\ObjFun^\QueryObjMean$ is measurable~\cite{Puterman-book}.

\begin{remark}
	There are several distinct definitions of \enquote{weighted reachability}.
	The one chosen here primarily serves as foundation for the more general mean payoff.
\end{remark}

\section{Introducing the Conditional Value-at-risk}

In order to define our problem, we first introduce the general concept of \emph{conditional value-at-risk} (CVaR), also known as \emph{average value-at-risk}, \emph{expected shortfall}, and \emph{expected tail loss}.
As already hinted, the CVaR of some real-valued random variable $X$ and probability $p \in [0, 1]$ intuitively is the expectation below the worst $p$-quantile of $X$.

Let $X : \ProbSpace \to \Reals$ be a random variable over the probability space $(\ProbSpace, \Salgebra, \Prob)$.
The associated \emph{cumulative density function} (CDF) $F_X : \Reals \to [0, 1]$ of $X$ yields the probability of $X$ being less than or equal to the given value $r$, i.e. $F_X(r) = \Prob(\{X(\event) \leq r\})$.
$F$ is non-decreasing and right continuous with left limits (\emph{càdlàg}).

The \emph{value-at-risk} $\VaR_p$ is the worst $p$-quantile, i.e.\ a value $v$ s.t. the probability of $X$ attaining a value less than or equal to $v$ is $p$:\footnote{An often used, mostly equivalent definition is $\inf \{r \in \Reals \mid F_X(r) \geq p\}$.
Unfortunately, this would lead to some complications later on.
See Sec.~\ref{sec:app:cvar_prop} for details.}
\begin{equation*}
	\VaR_p(X) := \sup \{r \in \Reals \mid F_X(r) \leq p\} \quad (\VaR_1(X) = \infty)
\end{equation*}
Then, with $v = \VaR_p(X)$, $\CVaR$ can be defined as~\cite{Rockafellar00optimizationof}
\begin{equation*}
	\CVaR_p(X) := \Expectation[X \mid X \leq v] = \frac{1}{p} \int_{(-\infty, v]} x \ dF_X,
\end{equation*}
with the corner cases $\CVaR_0 := \VaR_0$ and $\CVaR_1 = \Expectation$.

Unfortunately, this definition only works as intended for continuous $X$, as shown by the following example.

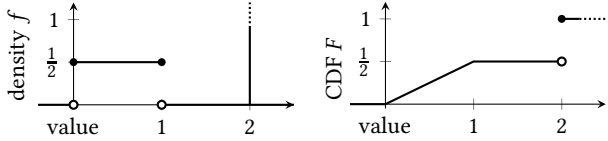
\begin{figure}
	\begin{subfigure}[!t]{0.49\columnwidth}
		\begin{tikzpicture}
			\begin{axis}[axis x line=middle, axis y line=middle,
					height=3cm, width=1.2\textwidth,
					x label style={at={(axis description cs:0,0)},anchor=north west},
					y label style={at={(axis description cs:0,0)},rotate=90,anchor=south west},
					xmin=-0.4,xmax=2.5,xtick={1,2},xlabel={value},
					ymin=-0.05,ymax=1.2,ytick={0.5,1},yticklabels={$\tfrac{1}{2}$, $1$},ylabel={density $f$}]

				\draw[thick] (axis cs:-1,0) -- (axis cs:0,0) node[discontinuity-limit] {};
				\draw[thick] (axis cs:0,0.5) node[discontinuity-value] {} -- (axis cs:1,0.5) node[discontinuity-value] {};
				\draw[thick] (axis cs:1,0) node[discontinuity-limit] {} -- (axis cs:3,0);
				\draw[thick] (axis cs:2,0) -- (axis cs:2,0.9);
				\draw[thick,densely dotted] (axis cs:2,0.9) -- (axis cs:2,1.2);
,			\end{axis}
		\end{tikzpicture}
	\end{subfigure}
	\begin{subfigure}[!t]{0.49\columnwidth}
		\begin{tikzpicture}
			\begin{axis}[axis x line=middle, axis y line=middle,
					height=3cm, width=1.2\textwidth,
					x label style={at={(axis description cs:0,0)},anchor=north west},
					y label style={at={(axis description cs:0,0)},rotate=90,anchor=south west},
					xmin=-0.4,xmax=2.5,xtick={1,2},xlabel={value},
					ymin=-0.05,ymax=1.2,ytick={0.5,1},yticklabels={$\tfrac{1}{2}$, $1$},ylabel={CDF $F$}]

				\draw[thick] (axis cs:-1,0) -- (axis cs:0,0) -- (axis cs:1,0.5) -- (axis cs:2,0.5) node[discontinuity-limit] {};
				\draw[thick] (axis cs:2,1) node[discontinuity-value] {} -- (axis cs:2.2,1);
				\draw[thick,densely dotted] (axis cs:2.2,1) -- (axis cs:2.5,1);

				%\draw[dashed] (axis cs:0,0.6) node[rotate=90,anchor=south] {$\tfrac{1}{2} + \varepsilon$} -- (axis cs:4,0.6);
			\end{axis}
		\end{tikzpicture}
	\end{subfigure}
	\caption{Distribution showing peculiarities of $\CVaR$}
	\label{fig:cvar_difficult}
\end{figure} %
\begin{example} \label{ex:cvar_intuition_wrong}
	Consider a random variable $X$ with a distribution as outlined in Fig.~\ref{fig:cvar_difficult}.
	For $p < \frac{1}{2}$, we certainly have $\VaR_p = 2p$.
	On the other hand, for any $p \in (\frac{1}{2}, 1)$, we get $\VaR_p = 2$.
	Consequently, the integral remains constant and $\CVaR_p$ would actually \emph{decrease} for increasing $p$, not matching the intuition. \QEE
\end{example}

%We present the general definition of CVaR, able to cope with any measurable random variable.
%These variables may exhibit both discrete and continuous behaviour, and thus we have to resort to a more measure-theoretic approach.
%
%\paragraph{Formal definition.}
%To avoid erratic behaviour, we assume that any random variables we deal with are bounded.
%This especially implies that the expectation of $X$ is defined.
%The \emph{density measure} of $X$ is given by $f : \Borel(\Reals) \to [0, 1]$,
%%
%\begin{equation*}
%	f(B) = \mu(X^{-1}(B)) = \measure(\{\event \in \ProbSpace \mid X(\event) \in B\}),
%\end{equation*}
%%
%where $\Borel(\Reals)$ is the Lebesgue-$\sigma$-field on $\Reals$.
%Essentially, $f$ describes how likely it is that $X$ attains some value in the given measurable set.
%From this measure we define the \emph{cumulative density function} (CDF) $F : \Reals \to [0, 1]$ which given $r$ yields the probability of $X$ being less than or equal to $r$, i.e.
%%
%\begin{equation*}
%	F(r) = f((-\infty, r]) = \measure(\{\event \in \ProbSpace \mid X(\event) \leq r\})
%\end{equation*}
%%
%$F$ is non-decreasing and right continuous with left limits (known as \emph{càdlàg}).
%

\paragraph{General definition.}
As seen in Ex.~\ref{ex:cvar_intuition_wrong}, the previous definition breaks down when $F_X$ is not continuous at the $p$-quantile and consequently $F_X(\VaR_p(X)) > p$.
Thus, we handle the values at the threshold separately, similar to~\cite{rockafellar2002conditional}.

\begin{definition}
	Let $X$ be some random variable and $p \in [0, 1]$.
	With $v = \VaR_p(X)$, the $\CVaR$ of $X$ is defined as
	\begin{equation*}
		\CVaR_p(X) := \frac{1}{p} \left( \int_{(-\infty, v)} x \, dF_X + (p - \Prob[X < v]) \cdot v \right),
	\end{equation*}
	which can be rewritten as
	\begin{equation*}
		\CVaR_p(X) = \tfrac{1}{p} \big( \Prob[X < v] \cdot \Expectation[X \mid X < v] + (p - \Prob[X < v]) \cdot v \big).
	\end{equation*}
	The corner cases again are $\CVaR_0 := \VaR_0$, and $\CVaR_1 = \Expectation$.
\end{definition}
Since the degenerate cases of $p = 0$ and $p = 1$ reduce to already known problems, we exclude them in the following.

We demonstrate this definition on the previous example.
\begin{example}
	Again, consider the random variable $X$ from Ex.~\ref{ex:cvar_intuition_wrong}.
	For $\frac{1}{2} < p < 1$ we have that $\Prob[X < \VaR_p(X)] = \Prob[X < 2] = \frac{1}{2}$.
	The right hand side of the definition $(p - \Prob[X < \VaR_p(X)]) = p - \frac{1}{2}$ captures the remaining discrete probability mass which we have to handle separately.
	Together with  $\int_{(-\infty, 2)} x \, dF_X = \frac{1}{4}$ we get $\CVaR_p(X) = \frac{1}{p} (\frac{1}{4} + (p - \frac{1}{2}) \cdot 2) = 2 - \frac{3}{4p}$.
	For example, with $p = \frac{3}{4}$, this yields the expected result $\CVaR_p(X) = 1$. \QEE
\end{example}
\begin{remark}
	Recall that $\Prob[X < r]$ can be expressed as the left limit of $F_X$, namely $\Prob[X < r] = \lim_{r' \to^{-} r} F_X(r')$.
	Hence, $\CVaR_p(X)$ solely depends on the CDF of $X$ and thus random variables with the same CDF also have the same CVaR.
\end{remark}

We say that \emph{$F_1$ stochastically dominates $F_2$} for two CDF $F_1$ and $F_2$, if $F_1(r) \leq F_2(r)$ for all $r$.
Intuitively, this means that a sample drawn from $F_2$ is likely to be larger or equal to a sample from $F_1$.
%Intuitively, this means that obtaining a larger value is more likely under $F_1$ than $F_2$.
All three investigated operators ($\Expectation$, $\CVaR$, and $\VaR$) are monotone w.r.t. stochastic dominance, see Sec.~\ref{sec:app:cvar_prop}.
\section{CVaR in MC and MDP: Problem statement}

Now, we are ready to define our problem framework.
First, we explain the types of building blocks for our queries, namely lower bounds on expectation, CVaR, and VaR.
Formally, we consider the following types of constraints.
\begin{align*}
	\QueryThreshExp & \leq \Expectation(X) & \QueryThreshCVaR & \leq \CVaR_\QueryProbCVaR(X) & \QueryThreshVaR & \leq \VaR_\QueryProbVaR(X)
\end{align*}
$X$ is some real-valued random variable, assigning a payoff to each run.
With these constraints, the classes of queries are denoted by
\begin{equation*}
	\QueryMDP^{\QueryCrit}_{\QueryObj, \QueryDim}
\end{equation*}
\begin{itemize}
	\item $\QueryCrit \subseteq \{\QueryCritExp, \QueryCritCVaR, \QueryCritVaR\}$ are the types of constraints,
	\item $\QueryObj \in \{\QueryObjReach, \QueryObjMean\}$ is the type of the objective function, either weighted reachability $\QueryObjReach$ or mean payoff $\QueryObjMean$, and
	\item $\QueryDim \in \{\QueryDimSingle, \QueryDimMulti\}$ is the dimensionality of the query.
\end{itemize}
We use $d$ to denote the dimensions of the problem, $d = 1$ iff $\QueryDim = \QueryDimSingle$.
As usual, we assume that all quantities of the input, e.g., probabilities of distributions, are rational.

An instance of these queries is specified by an MDP $\MDP$, a $d$-dimensional reward function $\reward : \States \to \Rationals^d$, and constraints from $\QueryCrit$, given by vectors $\QueryThreshExpVec, \QueryThreshCVaRVec, \QueryThreshVaRVec \in (\Rationals \union \{ \bot \})^d$ and $\QueryProbCVaRVec, \QueryProbVaRVec \in (0, 1)^d$.
This implies that in each dimension there is at most one constraint per type.
The presented methods can easily be extended to the more general setting of multiple constraints of a particular type in one dimension.
The decision problem is to determine whether there exists a strategy $\strategy$ such that \emph{all} constraints are met.

Technically, this is defined as follows.
Let $X$ be the $d$-dimensional random variable induced by the objective $\QueryObj$ and reward function $\reward$, operating on the probability space of $\MDP^\strategy$.
The strategy $\strategy$ is a witness to the query iff for each dimension $j \in [d]$ we have that $\Expectation[X_j] \geq \QueryVec{e}_j$, $\CVaR_{\QueryProbCVaRVec_j}(X_j) \geq \QueryThreshCVaRVec_j$, and $\VaR_{\QueryProbVaRVec_j}(X_j) \geq \QueryThreshVaRVec_j$. % where the three operators are computed with respect to the probability space of $\MDP^\strategy$.
Moreover, $\bot$ constraints are trivially satisfied.

For completeness sake, we also consider $\QueryMC^{\QueryCrit}_{\QueryObj, \QueryDim}$ queries, i.e. the corresponding problem on (finite state) Markov chains.

\paragraph{Notation.}
We introduce the following abbreviations.
When dealing with an MDP $\MDP$, $\CVaR^\strategy_p$ denotes $\CVaR_p$ relative to the probability space over runs induced by the strategy $\strategy$.
When additionally the random variable $X$ (e.g., mean payoff) is clear from the context, we may write $\CVaR_p$ and $\CVaR_p^\strategy$ instead of $\CVaR_p(X)$ and $\CVaR^\strategy_p(X)$, respectively.
We also define analogous abbreviations for $\VaR$.

\section{Single dimension} \label{sec:single}

We show that all queries in one dimension are in P.
Furthermore, our LP-based decision procedures directly yield a description of a witness strategy and allow for optimization objectives.
We refer to the input constraints by $\QueryThreshExp$ for expectation, $(\QueryProbCVaR, \QueryThreshCVaR)$ for CVaR, and $(\QueryProbVaR, \QueryThreshVaR)$ for VaR.
Further, we use $i$ for indices related to SCCs / MECs.

\subsection{Weighted reachability}

First, we show the simple result for Markov Chains, providing some insight in the techniques used in the MDP case.

\begin{theorem} \label{stm:mc_single_ptime}
	$\QueryMC^{\{\QueryCritExp, \QueryCritCVaR, \QueryCritVaR\}}_{\QueryObjReach, \QueryDimSingle}$ is in P.
\end{theorem}
\begin{proof}
	Let $\MC$ be a \emph{finite-state} Markov chain, $\reward$ a reward function, and $\reachSet = \{b_1, \dots, b_n\}$ the target set.
	Recall that all $b_i$ are absorbing, hence single-state BSCCs.
	We obtain the stationary distribution $p$ of $\MC$ in polynomial time by, e.g., solving a linear equation system~\cite{Puterman-book}.
	With $p$, we can directly compute the CDF of $\ObjFun^\QueryObjReach$ as $F_{\ObjFun^\QueryObjReach}(v) = \sum_{b_i : \reward(b_i) \leq v} p(b_i)$ and immediately decide the query.
	%This immediately gives us the exact expectation, VaR, and CVaR of $\MC$.
	%Since the transformation from $\MC$ to $\MC'$ preserves the distribution of $\ObjFun^\QueryObjReach$, and thus also the CDF, $\MC$ obtains the same values.
\end{proof}

Let us consider the more complex case of MDP.
We show a lower bound on the type of strategies necessary to realize $\QueryObj = \QueryObjReach$ queries with constraints on expectation and one of VaR or CVaR.
We then continue to prove that this class of strategies is optimal.
This characterization is used to derive a polynomial time decision procedure based on a linear program (LP) which immediately yields a witness strategy.
Finally, when we deal with the mean payoff case in Sec.~\ref{sec:single:mean}, we make use of the reasoning presented in this section.

\paragraph{Randomization is necessary for weighted reachability.}
In the following example, we present a simple MDP on which all deterministic strategies fail to satisfy specific constraints, while a straightforward randomizing one succeeds in doing so.

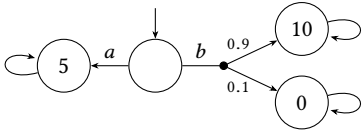
\begin{figure}
	\begin{tikzpicture}[auto,initial text={},node distance=0.5cm]
		\node[initial above,state] (s0) {};
		\node[state,left=0.5cm of s0] (s1) {$5$};
		\node[actionnode,right=of s0] (s0b) {};
		\node[state,above right=0.2cm and 0.75cm of s0b] (s2) {$10$};
		\node[state,below right=0.2cm and 0.75cm of s0b] (s3) {$0$};

		\path[->]
			(s0) edge[directedge,swap]    node[action] {$a$} (s1)
			(s0) edge[actionedge]         node[action] {$b$} (s0b)
			(s0b) edge[probedge]          node[prob] {$0.9$} (s2)
			(s0b) edge[probedge,swap]     node[prob] {$0.1$} (s3)
			(s1) edge[directedge,loop left] (s1)
			(s2) edge[directedge,loop right] (s2)
			(s3) edge[directedge,loop right] (s3);
	\end{tikzpicture}
	\caption{MDP used to show various difficulties of $\CVaR$}
	\label{fig:mdp_cvar_hard}
\end{figure}

\begin{example} \label{ex:weighted_reachability_randomization}
	Consider the MDP outlined in Fig.~\ref{fig:mdp_cvar_hard}.
	The only non-determinism is given by the choice in the initial state $s_0$.
	Hence, any strategy is characterised by the choice in that particular state.
	Let now $\strategy_a$ and $\strategy_b$ denote the deterministic strategies playing $a$ and $b$ in $s_0$, respectively.
	Clearly, $\strategy_a$ achieves an expectation, $\CVaR_{0.05}^{\strategy_a}$, and $\VaR_{0.05}^{\strategy_a}$ of $5$.
	On the other hand, $\strategy_b$ obtains an expectation of $9$ with $\CVaR_{0.05}^{\strategy_b}$ and $\VaR_{0.05}^{\strategy_b}$ equal to $0$.

	Thus, neither strategy satisfies the constraints $\QueryProbVaR = \QueryProbCVaR = 0.05$, $\QueryThreshExp = 6$, and $\QueryThreshCVaR = 2$ (or $\QueryThreshVaR = 5$).
	This is the case even when the strategy has arbitrary (deterministic) memory at its disposal, since in the first step there is nothing to remember.
	Yet, $\strategy = \frac{3}{4}\strategy_a + \frac{1}{4}\strategy_b$ achieves $\Expectation = \frac{3}{4}5 + \frac{1}{4}9 = 6 \geq \QueryThreshExp$, $\CVaR_\QueryProbCVaR = 2.5 \geq \QueryThreshCVaR$, and $\VaR_\QueryProbVaR = 5 \geq \QueryThreshVaR$. \QEE
\end{example}

Hence strategies satisfying an expectation constraint together with either a CVaR \emph{or} VaR constraint may necessarily involve randomization in general.
We prove that (i)~under mild assumptions randomization actually is sufficient, i.e. no memory is required, and (ii)~fixed memory may additionally be required in general.

\begin{definition} \label{def:attraction_assumption}
	Let $\MDP$ be an MDP with target set $\reachSet$ and reward function $\reward$.
	We say that $\MDP$ satisfies the \emph{attraction assumption} if \textbf{A1})~the target set $\reachSet$ is reached almost surely for any strategy, or \textbf{A2})~for all target state $s \in \reachSet$ we have $\reward(s) \geq 0$.
\end{definition}

Essentially, this definition implies that an optimal strategy never remains in a non-target MEC.
This allows us to design memoryless strategies for the weighted reachability problem.

\begin{theorem} \label{stm:mdp_exp_reach_strategies}
	Memoryless randomizing strategies are sufficient for $\QueryMDP_{\QueryObjReach, \QueryDimSingle}^{\{\QueryCritExp, \QueryCritVaR, \QueryCritCVaR\}}$ under the attraction assumption.
\end{theorem}

\begin{proof}
	Fix an MDP $\MDP$ and reward function $\reward$.
	We prove that for any strategy $\strategy$ there exists a memoryless, randomizing strategy $\strategy'$ achieving at least the expectation, VaR, and CVaR of $\strategy$.

	All target states $t_i \in \reachSet$ form single-state MECs, as we assumed that all target states are absorbing.
	Consequently, $\strategy$ naturally induces a distribution over these $s_i$.
	Now, we apply \cite[Theorem 3.2]{DBLP:journals/lmcs/EtessamiKVY08} to obtain a strategy $\strategy'$ with $\Prob^{\strategy'}[\LtlEventually s_i] \geq \Prob^\strategy[\LtlEventually s_i]$ for all $i$.

	With \textbf{A1}), we have $\sum p_i = 1$ and thus $\Prob^{\strategy'}[\LtlEventually t_i] = \Prob^\strategy[\LtlEventually t_i]$.
	Hence, $\strategy'$ obtains the same CDF for the weighted reachability objective.
	Under \textbf{A2}), the CDF $F'$ of strategy $\strategy'$ stochastically dominates the CDF $F$ of the original strategy $\strategy$, concluding the proof.
\end{proof}

\begin{theorem} \label{stm:mdp_exp_reach_strategies_general}
	Two-memory stochastic strategies (i.e.\ with both randomization and stochastic update) are sufficient for $\QueryMDP_{\QueryObjReach, \QueryDimSingle}^{\{\QueryCritExp, \QueryCritVaR, \QueryCritCVaR\}}$.
\end{theorem}

The proof is a simple application of the following Thm.~\ref{stm:mdp_exp_mean_strategies}, as weighted reachability is a special case of mean payoff.
Together with an example for the lower bound it can be found in Sec~\ref{sec:app:cvar:weighted_reach}.

\begin{figure}
	\begin{enumerate}
		\renewcommand{\labelenumi}{(\arabic{enumi}) }
		\item \label{fig:reach_lp:non-neg}
		All variables $y_a$, $x_s$, $\underline{x}_s$ are non-negative.
		\item \label{fig:reach_lp:flow}
		Transient flow for $s \in \States$:
		\begin{equation*}
			\mathbbm{1}_{s_0}(s) + {\sum}_{a \in \Actions} y_a \Trans(a, s) = {\sum}_{a \in \AvAct(s)} y_a + x_s
		\end{equation*}
		\item \label{fig:reach_lp:recurrent_switch}
		Switching to recurrent behaviour:
		\begin{equation*}
			{\sum}_{s \in \reachSet} x_s = 1
		\end{equation*}
		\item \label{fig:reach_lp:var_split}
		$\VaR$-consistent split:
		\begin{align*}
			\underline{x}_s & = x_s \text{ for $s \in \reachSet_<$} & \underline{x}_s & \leq x_s \text{ for $s \in \reachSet_=$}
		\end{align*}
		\item \label{fig:reach_lp:probability_split}
		Probability-consistent split:
		\begin{equation*}
			{\sum}_{s \in \reachSet_\leq} \underline{x}_s = \QueryProbCVaR
		\end{equation*}
		\item \label{fig:reach_lp:satisfaction}
		CVaR and expectation satisfaction:
		\begin{align*}
			{\sum}_{s \in \reachSet_\leq} \underline{x}_s \cdot \reward(s) & \geq \QueryProbCVaR \cdot \QueryThreshCVaR &
			{\sum}_{s \in \reachSet} x_s \cdot \reward(s) & \geq \QueryThreshExp
		\end{align*}
	\end{enumerate} %
	\caption{LP used to decide weighted reachability queries given a guess $t$ of $\VaR_\QueryProbCVaR$.
		$\reachSet_\sim := \{s \in \reachSet \mid \reward(s) \sim t\}$, $\sim \in \{<, =, \leq\}$.} %
	\label{fig:reach_lp} %
\end{figure} %
Inspired by~\cite[Fig.~3]{DBLP:conf/lics/ChatterjeeKK15}, we use the optimality result from Thm.~\ref{stm:mdp_exp_reach_strategies} to derive a decision procedure for weighted reachability queries under the attraction assumptions based on the LP in Fig.~\ref{fig:reach_lp}.
%There, given a value $t$, we use $\reachSet_\sim$ to refer to the set $\reachSet_\sim = \{s \in \reachSet \mid \reward(s) \sim t\}$ for $\sim \in \{<, =, \leq\}$.

To simplify the LP, we make further assumptions -- see Sec~\ref{sec:app:reach_lp_assumptions} for details.
First, all MECs, including non-target ones, consist of a single state. % and we identify each MEC with the corresponding state.
Second, all MECs from which $\reachSet$ is not reachable are considered part of $\reachSet$ and have $\reward = 0$ (similar to the \enquote{cleaned-up MDP} from \cite{DBLP:journals/lmcs/EtessamiKVY08}).
Finally, we assume that the quantile-probabilities are equal, i.e.\ $\QueryProbCVaR = \QueryProbVaR$.
The LP can easily be extended to account for different values by duplicating the $\underline{x}_s$ variables and adding according constraints.

The central idea is to characterize randomizing strategies by the \enquote{flow} they achieve.
To this end, Equality~\eqref{fig:reach_lp:flow} essentially models Kirchhoff's law, i.e.\ inflow and outflow of a state have to be equal.
In particular, $y_a$ expresses the transient flow of the strategy as the expected total number of uses of action $a$.
Similarly, $x_s$ models the recurrent flow, which under our absorption assumption equals the probability of reaching $s$.
Equality~\eqref{fig:reach_lp:recurrent_switch} ensures that all transient behaviour eventually changes into recurrent one.

In order to deal with our query constraints, Constraints~\eqref{fig:reach_lp:var_split} and \eqref{fig:reach_lp:probability_split} extract the worst $\QueryProbCVaR$ fraction of the recurrent flow, ensuring that the $\VaR_\QueryProbCVaR$ is at least $t$.
Note that equality is not guaranteed by the LP; if $\underline{x}_s = x_s$ for all $s \in \reachSet_\leq$, we have $\VaR_\QueryProbCVaR > t$.
Finally, Inequality~\eqref{fig:reach_lp:satisfaction} enforces satisfaction of the constraints.

\begin{theorem} \label{stm:mdp_exp_reach_lp_correct}
	Let $\MDP$ be an MDP with target states $\reachSet$ and reward function $\reward$, satisfying the attraction assumption.
	Fix the constraint probability $\QueryProbCVaR \in (0, 1)$ and thresholds $\QueryThreshExp, \QueryThreshCVaR \in \Rationals$.
	Then, we have that
	\begin{enumerate}
		\item
		for any strategy $\strategy$ satisfying the constraints, there is a $t \in \reward(\States)$ such that the LP in Fig.~\ref{fig:reach_lp} is feasible, and
		\item \label{stm:mdp_exp_reach_lp_correct:strategy}
		for any threshold $t \in \reward(\States)$, a solution of the LP in Fig.~\ref{fig:reach_lp} induces a memoryless, randomizing strategy $\strategy$ satisfying the constraints and $\VaR^\strategy_\QueryProbCVaR \geq t$.
	\end{enumerate}
\end{theorem}

\begin{proof}
	First, we prove for a strategy $\strategy$ satisfying the constraints that there exists a $t \in \reward(S)$ such that the LP is feasible. 
	By Thm.~\ref{stm:mdp_exp_reach_strategies}, we may assume that $\strategy$ is a memoryless randomizing strategy.
	From \cite[Theorem 3.2]{DBLP:journals/lmcs/EtessamiKVY08}, we get an assignment to the $y_a$'s and $x_s$'s satisfying Equalities~\eqref{fig:reach_lp:non-neg}, \eqref{fig:reach_lp:flow}, and~\eqref{fig:reach_lp:recurrent_switch} such that $\Prob^\strategy[\LtlEventually s] = x_s$ for all target states $s \in \reachSet$.
	Further, let $v = \VaR^\strategy_\QueryProbCVaR$ be the value-at-risk of the strategy.
	By definition of $\VaR$, we have that $\Prob^\strategy[X < v] \leq p$.

	Assume for now that $\Prob^\strategy[X < v] = p$, i.e. the probability of obtaining a value strictly smaller than $v$ is exactly $p$.
	In this case, choose $t$ to be the next smaller reward, i.e.\ $t = \max \{\reward(s) < v\}$.
	We set $\underline{x}_s = x_s$ for all $s \in \reachSet_\leq$, satisfying Constraints~\eqref{fig:reach_lp:var_split} and \eqref{fig:reach_lp:probability_split}.

	Otherwise, we have $\Prob^\strategy[X < v] < p$.
	Now, some non-zero fraction of the probability mass at $v$ contributes to the $\CVaR$.
	Again, we set the values for $\underline{x}_s$ according to Constraint~\eqref{fig:reach_lp:var_split}.
	The only degree of freedom are the values of $\underline{x}_s$ where $\reward(s) = t$.
	There, we assign the values so that $\sum_{s \in \reachSet_=} \underline{x}_s = p - \sum_{s \in \reachSet_<} \underline{x}_s$, satisfying Equality~\eqref{fig:reach_lp:probability_split}.

	It remains to check Inequality~\eqref{fig:reach_lp:satisfaction}.
	For expectation, we have $\sum_{s \in \reachSet} x_s \cdot \reward(s) = \sum_{s \in \reachSet} \Prob^\strategy[\LtlEventually s] \cdot \reward(s) = \Expectation^\strategy[\ObjFun^\QueryObjReach] \geq \QueryThreshExp$.
	For CVaR, notice that, due to the already proven Constraints~\eqref{fig:reach_lp:var_split} and \eqref{fig:reach_lp:probability_split}, the side of Inequality~\eqref{fig:reach_lp:satisfaction} is equal to $\CVaR_\QueryProbCVaR^\strategy$ and thus at least $\QueryThreshCVaR$.

	Second, we prove that a solution to the LP induces the desired strategy $\strategy$.
	Again by \cite[Theorem 3.2]{DBLP:journals/lmcs/EtessamiKVY08}, we get a memoryless randomizing strategy $\strategy$ such that $\Prob^\strategy[\LtlEventually s] = x_s$ for all states $s \in \reachSet$.
	Then $\Expectation^\strategy[\ObjFun^\QueryObjReach] = \sum_{s \in \reachSet} \Prob^\strategy[\LtlEventually s] \cdot \reward(s) = \sum_{s \in \reachSet} x_s \cdot \reward(s) \geq \QueryThreshExp$.
	Further,
	\begin{equation*}
		\CVaR_p(\ObjFun^\QueryObjReach) = \frac{1}{\QueryProbCVaR} \left( {\sum}_{s : \reward(s) < v} x_s \cdot \reward(s) + (\QueryProbCVaR - {\sum}_{s : \reward(s) < v} x_s) \cdot v \right)
	\end{equation*}
	by definition.
	Now, we make a case distinction on $\underline{x}_s = x_s$ for all $s \in \reachSet_=$.
	If this is true, we have $v = \VaR_\QueryProbCVaR^\strategy = \min\{r \in \reward(S) \mid r > t\}$, but $\Prob^\strategy[X < v] = p$.
	Consequently, $\reachSet_\leq = \{s \in \reachSet : \reward(s) < v\}$ and ${\sum}_{s : \reward(s) < v} x_s = p$.
	Otherwise, we have $v = t$ and consequently $\reachSet_< = \{s \mid \reward(s) < v\}$.
	Inserting in the above equation immediately gives the result $\CVaR_p(\ObjFun^\QueryObjReach) = \tfrac{1}{\QueryProbCVaR} {\sum}_{s \in \reachSet_\leq} \reward(s) \cdot \underline{x}_s$.
%		& = \tfrac{1}{\QueryProbCVaR} \left( {\sum}_{s \in \reachSet_<} \reward(s) \cdot \underline{x}_s + t \cdot {\sum}_{s \in \reachSet_=} \underline{x}_s \right) \\
%		& = \tfrac{1}{\QueryProbCVaR} {\sum}_{s \in \reachSet_<} \reward(s) \cdot \underline{x}_s \qedhere
\end{proof}

The linear program requires to know the $\VaR_p^\strategy$ beforehand, which in turn clearly depends on the chosen strategy.
Yet, there are only linearly many values the random variable $\ObjFun^\QueryObjReach$ attains. %and its distribution is discrete for any strategy.
Thus we can simply try to find a solution for all potential values of $\VaR_p^\strategy$, i.e.\ $\{r \in \reward(\States)\}$, yielding a polynomial time solution.

\begin{corollary} \label{stm:mdp_exp_reach_ptime}
	$\QueryMDP^{\{\QueryCritExp, \QueryCritVaR, \QueryCritCVaR\}}_{\QueryObjReach,\QueryDimSingle}$ is in P.
\end{corollary}

\begin{proof}
	Under the attraction assumption, this follows directly from Thm.~\ref{stm:mdp_exp_reach_lp_correct}.
	In general, the reduction to mean payoff used by Thm.~\ref{stm:mdp_exp_reach_strategies_general} and the respective result from Cor.~\ref{stm:mdp_mean_ptime} show the result.
	% and that there are $\cardinality{\reward(\States)} \leq \cardinality{\States}$ possible values for $\VaR_p^\strategy(\ObjFun^\QueryObjReach)$.
	%If there are multiple constraints, we modify the LP by adding separate \enquote{below} and \enquote{above} variables for each~\cite{APPENDIX}.
\end{proof}

\subsection{Mean payoff} \label{sec:single:mean}

In this section, we investigate the case of $\QueryObj = \QueryObjMean$.
Again, the construction for MC is considerably simple, yet instructive for the following MDP case.
\begin{theorem}
	$\QueryMC^{\{\QueryCritExp, \QueryCritVaR, \QueryCritCVaR\}}_{\QueryObjMean,\QueryDimSingle}$ is in P.
\end{theorem}
\begin{proof}[Proof sketch]
	For each BSCC $B_i$, we obtain its expected mean payoff $r_i = \Expectation[\ObjFun^\QueryObjMean \mid B_i]$ through, e.g., a linear equation system~\cite{Puterman-book}.
	Almost all runs in $B_i$ achieve this mean payoff and thus the corresponding random variable is discrete.
	We reduce the problem to weighted reachability by using the known reformulation
	\begin{equation*}
		\Prob[\ObjFun^\QueryObjMean = c] = {\sum}_{B_i : r_i = c} P[\LtlEventually B_i].
	\end{equation*}

	We replace each of these BSCCs by a representative $b_i$ to obtain $\MC'$.
	Define the set of target states $\reachSet = \{b_i\}$ and the reachability reward function $\reward'(b_i) = r_i$.
	By applying the approach of Thm.~\ref{stm:mc_single_ptime}, we obtain the expectation, $\VaR$, and $\CVaR$ for reachability in $\MC'$ which by construction coincides with the respective values for mean payoff in $\MC$.
\end{proof}

For the MDP case, recall that simple expectation maximization of mean payoff can be reduced to weighted reachability~\cite{DBLP:conf/cav/AshokCDKM17} and deterministic, memoryless strategies are optimal~\cite{Puterman-book}.
Yet, solving a conjunctive query involving either VaR or CVaR needs more powerful strategies than in the weighted reachability case of Thm.~\ref{stm:mdp_exp_reach_strategies}.
Nevertheless, we show how to decide these queries in P.

\paragraph{Randomization and memory is necessary for mean payoff.}
A simple modification of the MDP in Fig.~\ref{fig:mdp_cvar_hard} yields an MDP where both randomization and memory is required to satisfy the constraints of the following example.

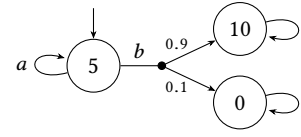
\begin{figure}
	\begin{tikzpicture}[auto,initial text={},node distance=0.5cm]
		\node[initial above,state] (s0) {$5$};
		\node[actionnode,right=of s0] (s0b) {};
		\node[state,above right=0.2cm and 0.75cm of s0b] (s2) {$10$};
		\node[state,below right=0.2cm and 0.75cm of s0b] (s3) {$0$};

		\path[->]
			(s0) edge[directedge,loop left] node[action] {$a$} (s0)
			(s0) edge[actionedge]         node[action] {$b$} (s0b)
			(s0b) edge[probedge]          node[prob] {$0.9$} (s2)
			(s0b) edge[probedge,swap]     node[prob] {$0.1$} (s3)
			(s2) edge[directedge,loop right] (s2)
			(s3) edge[directedge,loop right] (s3);
	\end{tikzpicture}
	\caption{Memory is necessary for mean payoff queries}
	\label{fig:mean_payoff_memory}
\end{figure} %

\begin{example} \label{ex:mdp_exp_reach_memory}
	Consider the MDP presented in Fig.~\ref{fig:mean_payoff_memory}.
	There, the same constraints as before, i.e.\ $\QueryProbVaR = \QueryProbCVaR = 0.05$, $\QueryThreshExp = 6$, and $\QueryThreshCVaR = 2$ (or $\QueryThreshVaR = 5$), can only be satisfied by strategies with both memory and randomization.
	Clearly, a pure strategy can only satisfy either of the two constraints again.
	But now a memoryless randomizing strategy also is insufficient, too, since any non-zero probability on action $b$ leads to almost all runs ending up on the right side of the MDP, hence yielding a $\CVaR_\QueryProbCVaR$ and $\VaR_\QueryProbVaR$ of $0$.
	Instead, a stochastic strategy with $\Memory = \{a, b\}$ can simply choose $\alpha = \{a \mapsto \frac{3}{4}, b \mapsto \frac{1}{4}\}$ and play the corresponding action indefinitely, satisfying the constraints. \QEE
\end{example}

We prove that this bound actually is tight, i.e.\ that, given stochastic memory update, two memory elements %and a single update
are sufficient.

\begin{theorem} \label{stm:mdp_exp_mean_strategies}
	Two-memory stochastic strategies (i.e.\ with both randomization and stochastic update) are sufficient for $\QueryMDP_{\QueryObjMean,\QueryDimSingle}^{\{\QueryCritExp, \QueryCritVaR, \QueryCritCVaR\}}$.
\end{theorem}

\begin{proof}
	%This construction is inspired by that of \cite[Lemma 21]{DBLP:conf/lics/ChatterjeeKK15}.
	Let $\strategy$ be a strategy on an MDP $\MDP$ with reward function $\reward$.
	We construct a two-memory stochastic strategy $\strategy'$ achieving at least the expectation, VaR, and CVaR of $\strategy$.

	First, we obtain a memoryless deterministic strategy $\strategy_\text{opt}$ which obtains the maximal possible mean payoff in each MEC~\cite{Puterman-book}.
	We then apply the construction of \cite[Proposition 4.2]{DBLP:journals/corr/abs-1104-3489} (see also \cite[Lemma 5.7]{DBLP:conf/lics/ChatterjeeKK15}), where the $\xi$ is our $\strategy_\text{opt}$.
	(Technically, this can be ensured by choosing the constraints of the LP $L$ according to $\strategy_\text{opt}$.)

	Intuitively, this constructs a two-memory strategy $\strategy'$ on $\MDP$ behaving as follows.
	Initially, $\strategy'$ remains in each MEC with the same probability as $\strategy$, i.e. $\Prob^{\strategy'}[\LtlEventually \LtlAlways M_i] = \Prob^{\strategy}[\LtlEventually \LtlAlways M_i]$ by following a memoryless \enquote{searching} strategy and stochastically switching its memory state to \enquote{remain}.
	Once in the \enquote{remain} state, the behaviour of the optimal strategy $\strategy_\text{opt}$ is implemented.

	Clearly, (i)~both strategies remain in a particular MEC with the same probability, and (ii)~$\strategy'$ obtains as least as much value in each MEC as $\strategy$.
	Hence the CDF induced by $\strategy'$ stochastically dominates the one of $\strategy$, concluding the proof.
\end{proof}

This immediately gives us a polynomial time decision procedure.

\begin{corollary} \label{stm:mdp_mean_ptime}
	$\QueryMDP^{\{\QueryCritExp, \QueryCritVaR, \QueryCritCVaR\}}_{\QueryObjMean, \QueryDimSingle}$ is in P.
\end{corollary}

Furthermore, we can use results of \cite[Lemma 16]{DBLP:conf/lics/ChatterjeeKK15} to trade the stochastic update for more memory.

\begin{corollary}
	Stochastic strategies with finite, deterministic memory are sufficient for $\QueryMDP_{\QueryObjMean,\QueryDimSingle}^{\{\QueryCritExp, \QueryCritVaR, \QueryCritCVaR\}}$.
\end{corollary}

\paragraph{Deterministic strategies may require exponential memory.}
As sources of randomness are not always available, one might ask what can be hoped for when only determinism is allowed.
As already shown in Ex.~\ref{ex:weighted_reachability_randomization}, randomization is required in general.
But even if some deterministic strategy is sufficient, it may require memory exponential in the size of the input, even in an MDP with only 3 states.
We show this in the following example. % by another variation of the MDP in Fig.~\ref{fig:mdp_cvar_hard}.

\begin{figure}
	\begin{tikzpicture}[auto,initial text={},node distance=0.5cm]
		\node[initial above,state] (s0) {$5$};
		\node[actionnode,right=1cm of s0] (s0b) {};
		\node[state,above right=0.2cm and 1cm of s0b] (s2) {$10$};
		\node[state,below right=0.2cm and 1cm of s0b] (s3) {$0$};

		\path[->]
			(s0) edge[directedge,loop left] node[action] {$a$} (s0)
			(s0) edge[actionedge]         node[action] {$b$} (s0b)
			(s0b) edge[probedge]          node[prob] {$0.9 \varepsilon$} (s2)
			(s0b) edge[probedge,swap]     node[prob] {$0.1 \varepsilon$} (s3)
			(s0b) edge[probedge,bend left] node[prob,pos=.6] {$1-\varepsilon$} (s0)
			(s2) edge[directedge,loop right] (s2)
			(s3) edge[directedge,loop right] (s3);
	\end{tikzpicture}
	\caption{Exponential memory is necessary for mean payoff when only deterministic update is allowed.}
	\label{fig:mean_payoff_deterministic_memory}
\end{figure}
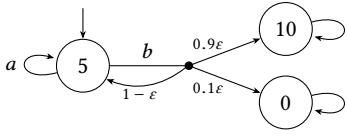

\begin{example}
	Consider the MDP outlined in Fig.~\ref{fig:mean_payoff_deterministic_memory} together with the constraints $\QueryProbVaR = \QueryProbCVaR = 0.05$, $\QueryThreshExp = 6$, and $\QueryThreshCVaR = 2$ (or $\QueryThreshVaR = 5$).
	Again, any optimal strategy needs a significant part of runs to go to the right side in order to satisfy the expectation constraint.
	Yet, any strategy can only \enquote{move} a small fraction of the runs there in each step.
	In particular, after $k$ steps, the right side is only reached with probability at most $1 - (1 - \varepsilon)^k$.
	When choosing $\varepsilon = 2^{-n}$, which needs $\Theta(n)$ bits to encode, a deterministic strategy requires $k \geq c / \log(1-2^{-n}) \in O(2^n)$ memory elements to count the number of steps.
	The same holds true for any deterministic-update strategy.

	On the other hand, a strategy with stochastic memory update can encode this counting by switching its state with a small probability after each step.
	For example, a strategy switching with probability $p = 3 \varepsilon$ from \enquote{play $b$} to \enquote{play $a$} satisfies the constraint. \QEE
\end{example}

\subsection{Single constraint queries}

In this section, we discuss an important sub-case of the single-dimensional case, namely queries with only a single constraint, i.e. $\cardinality{\QueryCrit} = 1$.
We show that deterministic memoryless strategies are sufficient in this case.

One might be tempted to use standard arguments and directly conclude this from the results of Thm.~\ref{stm:mdp_exp_reach_strategies} as follows.
Recall that this theorem shows that memoryless, randomizing strategies are sufficient; and that any such strategy can be written as finite convex combination of memoryless, deterministic strategies.
Most constraints, for example expectation or reachability, behave linearly under convex combination of strategies, e.g., $\Expectation^{\strategy_\lambda}(X) = \lambda \Expectation^{\strategy_1}[X] + (1-\lambda) \Expectation^{\strategy_2}[X]$.
%As a consequence, if two strategies $\strategy_1$ and $\strategy_2$ satisfy a constraint, any combination of them satisfies the constraints, too.
Consequently, for an optimal memoryless strategy, there is a deterministic witness, which in turn also is optimal.

Surprisingly, this assumption is not true for $\CVaR$.
On the contrary, the $\CVaR$ of a convex combination of strategies might be strictly worse than the $\CVaR$s of either strategy, as shown in the following example.
We prove a slightly weaker property of $\CVaR$ which eventually allows us to apply similar reasoning.

\begin{example} \label{ex:mdp_cvar_hard}
	% Note: Intentionally it's p not \QueryProbCVaR, since this is not decision context
	Recall the MDP in Fig.~\ref{fig:mdp_cvar_hard} and let $p = 0.05$.
	As previously shown, $\CVaR_p^{\strategy_a} = 5$ and $\CVaR_p^{\strategy_b} = 0$, but the mixed strategy $\strategy_\lambda = \frac{1}{2}\strategy_a + \frac{1}{2}\strategy_b$ achieves $\CVaR_p^{\strategy_\lambda} = 0$ instead of the convex combination $\frac{1}{2}5 + \frac{1}{2}0 = 2.5$.

	For $p = 0.2$, we have $\CVaR_p^{\strategy_a} = \CVaR_p^{\strategy_b} = 5$.
	Yet, \emph{any} non-trivial convex combination of the two strategies yields a $\CVaR_p$ less than $5$.
	See Sec.~\ref{sec:app:cvar:nonlinear} for more details.
	With according constraints, this effectively can force an optimal strategy to choose between $a$ or $b$.
	This observation is further exploited in the NP-hardness proof of the multi-dimensional case in Sec.~\ref{sec:multi}.
	\qee
\end{example}

Since CVaR considers the worst events, the CVaR of a combination intuitively cannot be better than the combination of the respective CVaRs.
We prove this intuition in the general setting, where instead of a convex combination of strategies we consider a mixture of two random variables.

\begin{lemma} \label{stm:cvar_quasiconvex}
	$\CVaR_p(X)$ is convex in $X$ for fixed $p \in (0, 1)$, i.e.\ for random variables $X_1, X_2$ and $\lambda \in [0, 1]$
	\begin{equation*}
		\CVaR_p(\lambda X_1 + (1-\lambda) X_2) \leq \lambda \CVaR_p(X_1) + (1-\lambda) \CVaR_p(X_2).
	\end{equation*}
\end{lemma}

The proof can be found in Sec.~\ref{sec:app:quasiconvex}.
This result allows us to apply the ideas outlined in the beginning of the section.

\begin{theorem}
	For any $\QueryObj\in\{\QueryObjReach,\QueryObjMean\}$, deterministic memoryless strategies are sufficient for $\QueryMDP^{\QueryCrit}_{\QueryObj, \QueryDimSingle}$ when $\cardinality{\QueryCrit} = 1$.
\end{theorem}

\begin{proof}
	This is known for $\QueryCrit = \{\QueryCritExp\}$~\cite{Puterman-book} and $\QueryCrit = \{\QueryCritVaR\}$~\cite{CIS-125853}.

	For CVaR, observe that the convex combination of deterministic strategies cannot achieve a better CVaR than the best strategy involved in the combination (see Lem.~\ref{stm:cvar_quasiconvex}).
	This immediately yields the result for $\QueryObj = \QueryObjReach$ through Thm.~\ref{stm:mdp_exp_reach_strategies}.
	For $\QueryObj = \QueryObjMean$, we exploit the approach of Thm.~\ref{stm:mdp_exp_mean_strategies}.
	Recall that there we obtained a two-memory strategy $\strategy'$.
	Both randomization and stochastic update are used solely to distribute the runs over all MECs accordingly.
	By the above reasoning, for each MEC it is sufficient to either almost surely remain there or leave it.
	This behaviour can be implemented by a deterministic memoryless strategy on the original MDP.
\end{proof}

\section{Multiple Dimensions} \label{sec:multi}

In this section, we deal with multi-dimensional queries. % i.e. $\QueryMDP^{\QueryCrit}_{\QueryObj,\QueryDimMulti}$
We continue to use $i$ for indices related to MECs and further use $j$ for dimension indices.
%Specifically, we show complexity results for multiple types of queries.
First, we show that the Markov Chain case does not significantly change.

\begin{theorem}
	For any $\QueryObj\in\{\QueryObjReach,\QueryObjMean\}$, $\QueryMC^{\{\QueryCritExp, \QueryCritVaR, \QueryCritCVaR\}}_{\QueryObj, \QueryDimMulti}$ is in P.
\end{theorem}

\begin{proof}
	Similarly to the single-dimensional case, we decide each constraint in each dimension separately, using our previous results.
	The query is satisfied iff each of the constraints is satisfied.
\end{proof}

\subsection{NP-Hardness of reachability and mean payoff}

For the MDP on the other hand, multiple dimensions add significant complexity.
In the following, we show that already the weighted reachability problem with multiple dimensions and only $\CVaR$ constraints, i.e.\ $\QueryMDP_{\QueryObjReach, \QueryDimMulti}^{\{\QueryCritCVaR\}}$, is NP-hard.
This result directly transfers to mean payoff, i.e.\ $\QueryObj = \QueryObjMean$.
Recall that in contrast $\QueryMDP_{\QueryObjReach, \QueryDimMulti}^{\{\QueryCritExp\}}$ and even $\QueryMDP_{\QueryObjReach, \QueryDimMulti}^{\{\QueryCritExp,\QueryCritVaR_0\}}$, i.e. constraints on the expectation and ensuring that almost all runs achieve a given threshold, are in P~\cite{DBLP:conf/lics/ChatterjeeKK15}.

\begin{theorem} \label{stm:reach_multi_np_hard}
	For any $\QueryObj \in \{\QueryObjReach, \QueryObjMean\}$, $\QueryMDP_{\QueryObj, \QueryDimMulti}^{\{\QueryCritCVaR\}}$ is NP-hard (when the dimension $d$ is a part of the input).
\end{theorem}

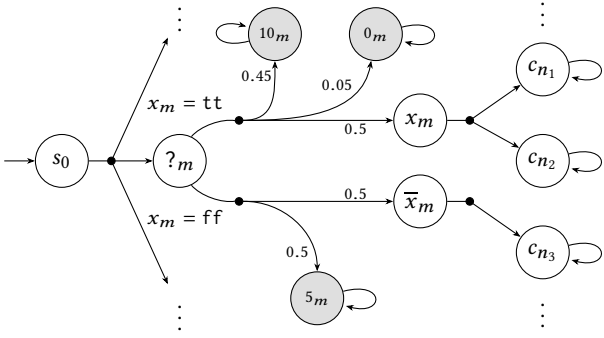
\begin{figure}
	\begin{tikzpicture}[auto,initial text=,node distance=0.5cm]
		\node[initial,state] (s0) {$\initstate$};

		\node[actionnode,right=0.25cm of s0] (s0a) {};

		\node[state,right=0.5cm of s0a] (x) {$\texttt{?}_m$};
		\node[draw=none,above=1.25cm of x] (xdotsabove) {$\vdots$};
		\node[draw=none,below=1.25cm of x] (xdotsbelow) {$\vdots$};

		\node[actionnode,above right=0.25cm and 0.5cm of x] (xt) {};
		\node[actionnode,below right=0.25cm and 0.5cm of x] (xf) {};

		\node[above right=1cm and 1cm of xt] (xt_commit_mid) {};
		\node[state,left=0.2cm of xt_commit_mid,fill=lightgray!50] (xt_commit_above) {\scriptsize $10_m$};
		\node[state,right=0.2cm of xt_commit_mid,fill=lightgray!50] (xt_commit_below) {\scriptsize $0_m$};
		\node[state,below right=1cm and 0.75cm of xf,fill=lightgray!50] (xf_commit) {\scriptsize $5_m$};

		\node[state,right=2cm of xt] (xt_commit_node) {$x_m$};
		\node[state,right=2cm of xf] (xf_commit_node) {$\overline{x}_m$};

		\node[actionnode,right=0.25cm of xt_commit_node] (xt_action) {};
		\node[actionnode,right=0.25cm of xf_commit_node] (xf_action) {};

		\node (clause_anchor) at ($(xt_action)!0.5!(xf_action)$) {};
		\node[right=0.5cm of clause_anchor,state] (c2) {$c_{n_2}$};
		\node[above=of c2,state]       (c1) {$c_{n_1}$};
		\node[below=of c2,state]       (c3) {$c_{n_3}$};

		\node[above=0.1cm of c1] {$\vdots$};
		\node[below=0cm of c3] {$\vdots$};
		
		\path[->]
			(s0) edge[actionedge]              node[action] {} (s0a)

			(s0a) edge[probedge,pos=0.75]      node[prob] {} (x)
			(s0a) edge[probedge,pos=0.75]      node[prob] {} (xdotsabove)
			(s0a) edge[probedge,pos=0.75,swap] node[prob] {} (xdotsbelow)

			(x) edge[actionedge,pos=0.9,looseness=0.75,bend left]       node[action,anchor=south east] {$x_m = \texttt{tt}$} (xt)
			(x) edge[actionedge,pos=0.9,looseness=0.75,bend right,swap] node[action,anchor=north east] {$x_m = \texttt{ff}$} (xf)

			(xt) edge[probedge,looseness=1.25,out=0,in=-90,pos=0.75]   node[prob] {$0.45$} (xt_commit_above)
			(xt) edge[probedge,looseness=1,out=0,in=-100,pos=0.75]     node[prob] {$0.05$} (xt_commit_below)
			(xt) edge[probedge,pos=0.75,swap]                          node[prob] {$0.5$}  (xt_commit_node)
			(xf) edge[probedge,looseness=1,out=1,in=90,swap,pos=0.75]  node[prob] {$0.5$}  (xf_commit)
			(xf) edge[probedge,pos=0.75]                               node[prob] {$0.5$}  (xf_commit_node)

			(xt_commit_above) edge[directedge,loop left] (xt_commit_above)
			(xt_commit_below) edge[directedge,loop right] (xt_commit_below)
			(xf_commit) edge[directedge,loop right] (xf_commit)

			(xf_commit_node) edge[actionedge] (xf_action)
			(xt_commit_node) edge[actionedge] (xt_action)

			(xt_action) edge[probedge] (c1)
			(xt_action) edge[probedge] (c2)
			(xf_action) edge[probedge] (c3)

			(c1) edge[directedge,loop right] (c1)
			(c2) edge[directedge,loop right] (c2)
			(c3) edge[directedge,loop right] (c3)
		;

%		\begin{scope}[on background layer]
%			\node[fill=lightgray!50,rectangle,rounded corners=5pt,fit=(xt_commit_above) (xt_commit_below) (xf_commit)] {};
%		\end{scope}
	\end{tikzpicture}
	\caption{Gadget for variable $x_m$.
	Uniform transition probabilities are omitted for readability.}
	\label{fig:reach_np_hard_gadget}
\end{figure}

\begin{proof}%[Proof sketch]
	We prove hardness by reduction from 3-SAT.
	The core idea is to utilize observations from Fig.~\ref{fig:mdp_cvar_hard} and Ex.~\ref{ex:mdp_cvar_hard}, namely that CVaR constraints can be used to enforce a deterministic choice.

	Let $\{C_n\}$ be a set of $N$ clauses with $M$ variables $x_m$ and set the dimensions $d = N + M$.
	By abuse of notation, $n$ refers to the dimension of clause $C_n$ and $m$ to the one of variable $x_m$, respectively.

	The gadget for the reduction is outlined in Fig.~\ref{fig:reach_np_hard_gadget}.
	Observe that, due to the structure of the MDP, we have that $\ObjFun^\QueryObjReach = \ObjFun^\QueryObjMean$.

	Overall, the reduction works as follows.
	Initially, a state $\texttt{?}_m$, representing the variable $x_m$, is chosen uniformly.
	In this state, the strategy is asked to give the valuation of $x_m$ through the actions \enquote{$x_m = \texttt{tt}$} or \enquote{$x_m = \texttt{ff}$}.
	As seen in Ex.~\ref{ex:mdp_cvar_hard}, the structure of the shaded states can be used to enforces a deterministic choice between the two actions.
	Particularly, in dimension $m$ we require $\CVaR_p \geq 5$ for $p = \frac{M-1}{M} + \frac{1}{M} \cdot 0.5 \cdot 0.2$.
	Since all other gadgets yield 0 in dimension $m$ and only half of the runs going through $\texttt{?}_m$ end up in the shaded area, this corresponds to Ex.~\ref{ex:mdp_cvar_hard}, where $p = 0.2$.

	Once in either state $x_m$ or $\overline{x}_m$, a state $c_n$ corresponding to a clause $C_n$ satisfied by this assignment is chosen uniformly.
	In the example gadget, we would have $x_m \in C_{n_1} \intersection C_{n_2}$, and $\overline{x}_m \in C_{n_3}$.
	We set the reward of $c_n$ to $1_n$.
	Then a clause $c_n$ is satisfied under the assignment if the state $c_n$ is visited with positive probability, e.g. if $\CVaR_1 \geq \frac{1}{M} \cdot 0.5 \cdot \frac{1}{N}$.
	Clearly, a satisfying assignment exists iff a strategy satisfying these constraints exists.
\end{proof}

\subsection{NP-completeness and strategies for reachability}

For weighted reachability, we prove that the previously presented bound is tight, i.e.\ that the weighted reachability problem with multiple dimensions and $\CVaR$ constraints is NP-complete when $d$ is part of the input and $P$ otherwise.
First, we show that the strategy bounds of the single dimensional case directly transfer.
Intuitively, this is the case since only the steady state distribution over the target set $\reachSet$ is relevant, independent of the dimensionality.

\begin{theorem}
	Two-memory stochastic strategies (i.e.\ with both randomization and stochastic update) are sufficient for $\QueryMDP_{\QueryObjReach, \QueryDimMulti}^{\{\QueryCritExp, \QueryCritVaR, \QueryCritCVaR\}}$.
	Moreover, if $\reward_j(s) \geq 0$ for all $s \in \reachSet$ and $j \in [d]$, then memoryless randomizing strategies are sufficient.
\end{theorem}

\begin{proof}
	Follows directly from the reasoning used in the proofs of Thm.~\ref{stm:mdp_exp_mean_strategies} and Thm.~\ref{stm:mdp_exp_reach_strategies}.
\end{proof}

\begin{figure}
	\begin{enumerate}
		\renewcommand{\labelenumi}{(\arabic{enumi}) }
		\item \label{fig:reach_lp_multi:non-neg}
		All variables $y_a$, $x_s$, $\underline{x}_s^j$ are non-negative.
		\setcounter{enumi}{3}
		\item \label{fig:reach_lp_multi:var_split}
		$\VaR$-consistent split for $j \in [d]$:
		\begin{align*}
			\underline{x}_s^j & = x_s \text{ for $s \in \reachSet_<^j$} & \underline{x}_s^j & \leq x_s \text{ for $s \in \reachSet_=^j$}
		\end{align*}
		\item \label{fig:reach_lp_multi:probability_split}
		Probability-consistent split for $j \in [d]$:
		\begin{equation*}
			{\sum}_{s \in \reachSet_\leq^j} \underline{x}_s^j = \QueryProbCVaRVec_j
		\end{equation*}
		\item \label{fig:reach_lp_multi:satisfaction}
		CVaR and expectation satisfaction for $j \in [d]$:
		\begin{align*}
			{\sum}_{s \in \reachSet_\leq^j} \underline{x}_s \cdot \reward(s) & \geq \QueryProbCVaR_j \cdot \QueryThreshCVaRVec_j &
			{\sum}_{s \in \reachSet} x_s \cdot \reward_j(s) & \geq \QueryThreshExpVec_j
		\end{align*}
	\end{enumerate}
	\caption{LP used to decide multi-dimensional weighted reachability queries given a guess $\textbf{t}$ of $\VaR_{\QueryProbCVaR_j}$.
		Equalities~\eqref{fig:reach_lp:flow} and \eqref{fig:reach_lp:recurrent_switch} are as in Fig.~\ref{fig:reach_lp}, $\reachSet_\sim^j := \{s \in \reachSet \mid \reward_j(s) \sim \textbf{t}_j\}$, $\sim \in \{<, =, \leq\}$.}
	\label{fig:reach_lp_multi}
\end{figure}

\begin{theorem}
	$\QueryMDP_{\QueryObjReach, \QueryDimMulti}^{ \{\QueryCritExp, \QueryCritVaR, \QueryCritCVaR\}}$ is in NP if $d$ is a part of the input; moreover, it is in P for any fixed $d$.
\end{theorem}

\begin{proof}[Proof sketch]
	To prove containment, we guess the VaR threshold vector $\mathbf{t}$ out of the set of potential ones, namely $\{r \mid \exists i \in [d], s \in \reachSet. \reward_i(s) = r\}^d$ and use an LP to verify the solution.
	We again assume that each MEC can reach the target set and is single-state, as we did for Fig.~\ref{fig:reach_lp}.
	The arguments used to resolve this assumption are still applicable in the multi-dimensional setting.
	The LP consists of the flow Equalities~\eqref{fig:reach_lp:flow} and \eqref{fig:reach_lp:recurrent_switch} from the LP in Fig.~\ref{fig:reach_lp} together with the modified (In)Equalities~\eqref{fig:reach_lp_multi:var_split}-\eqref{fig:reach_lp_multi:satisfaction} as shown in Fig.~\ref{fig:reach_lp_multi}.

	The difference is that we extract the worst fraction of the flow \emph{in each dimension}.
	Consequently, we have $d$ instances of each $\underline{x}_s$ variable, namely $\underline{x}_s^j$.
	The number of possible guesses $\mathbf{t}$ is bounded by $\cardinality{\reachSet}^d$ and thus the guess is of polynomial length.
	For a fixed $d$ the bound itself is polynomial and hence, as previously, we can try out all vectors.
\end{proof}

\subsection{Upper bounds of mean payoff}

In this section, we provide an upper bound on the complexity of mean-payoff queries.
Strategies in this context are known to have higher complexity.

\begin{proposition}[\cite{DBLP:journals/corr/abs-1104-3489}]
	Infinite memory is necessary for $\QueryMDP^{\{\QueryCritExp\}}_{\QueryObjMean, \QueryDimMulti}$.
\end{proposition} %
Note that this directly transfers to $\QueryMDP^{\{\QueryCritCVaR\}}_{\QueryObjMean, \QueryDimMulti}$, as $\CVaR_1 = \Expectation$.
However, closing gaps between lower and upper bounds for the mean payoff objective is notoriously more difficult.
For instance, $\QueryMDP_{\QueryObjMean,\QueryDimMulti}^{\{\QueryCritVaR\}}$ is known to be in EXP, but not even known to be NP-hard, neither is $\QueryMDP_{\QueryObjMean,\QueryDimMulti}^{\{\QueryCritExp,\QueryCritVaR\}}$.
Since we have proven that $\QueryMDP_{\QueryObjMean,\QueryDimMulti}^{\{\QueryCritCVaR\}}$ is NP-hard, we can expect that obtaining the matching NP upper bound will be yet more difficult.
The fundamental difference of the multi-dimensional mean-payoff case is that the solutions within MECs cannot be pre-computed, rather non-trivial trade-offs must be considered.
Moreover, the trade-offs are not \enquote{local} and must be synchronized over all the MECs, see~\cite{DBLP:conf/lics/ChatterjeeKK15} for details.

We now observe that, as opposed to quantile queries, i.e. $\VaR$ constraints, the behaviour inside each MEC can be assumed to be quite simple.
Our results primarily rely on \cite{DBLP:journals/lmcs/ChatterjeeKK17} and use a similar notation.
In particular, given a run $\run$, $\vfreq_a(\run)$ yields the average frequency of action $a$, i.e. $\vfreq_a(\run) := \liminf_{n \to \infty} \frac{1}{n} \sum_{t = 1}^n \mathbbm{1}_a(a_t)$, where $a_t$ refers to the action taken by $\run$ in step $t$.

\begin{definition}
	A strategy $\strategy$ is \emph{MEC-constant} if for all $M_i \in \MECs$ with $\Prob^\strategy[\LtlEventually\LtlAlways M_i] > 0$ and all $j \in [d]$ there is a $v \in \Reals$ such that
%	$\Probability[\ObjFun^\QueryObjMean_j(\rho)=v\mid \rho \text{ ends in }C]=1$.
	$\Prob^\strategy[\ObjFun^\QueryObjMean_j = v \mid \LtlEventually\LtlAlways M_i] = 1$.
\end{definition}

\begin{lemma} \label{stm:mean_multi_mec_constant_strategies}
	MEC-constant strategies are sufficient for $\QueryMDP^{\{\QueryCritExp, \QueryCritCVaR\}}_{\QueryObjMean, \QueryDimMulti}$.
\end{lemma}
\begin{proof}
	Fix an MDP $\MDP$ with MECs $\MECs = \{M_1, \dots, M_n\}$, reward function $\reward$ and a strategy $\strategy$.
	Further, define $p_i = \Prob^\strategy[\LtlEventually\LtlAlways M_i]$.
	We construct a strategy $\strategy'$ so that (i)~$\Prob^{\strategy'}[\LtlEventually\LtlAlways M_i] = p_i$ for all $M_i$, and (ii)~all behaviours of $\strategy$ on a MEC $M_i$ are \enquote{mixed} into each run on $M_i$, making it MEC-constant.

	We first define the mixing strategies $\strategy_i$, achieving point (ii).
	By \cite[Sec.~4.1]{DBLP:journals/lmcs/ChatterjeeKK17}, there are frequencies $(x_a)_{a \in \Actions}$ which
	\begin{itemize}
		\item satisfy $\sum_{a \in A} x_a \cdot \Trans(a, s) = \sum_{a \in \AvAct(s)} x_a$ for all $s \in \States$,
		\item for each action $a$ we have $\Expectation^\strategy[\vfreq_a] \leq x_a$, and
		\item $\sum_{a \in \Actions \intersection M_i} x_a = p_i$.
	\end{itemize}
	By \cite[Cor.~5.5]{DBLP:journals/lmcs/ChatterjeeKK17}, there is a (Markov) strategy $\strategy_i$ on $M_i$ where
	\begin{equation*}
		\Prob^{\strategy_i}\left[\vfreq_a = x_a / p_i \right] = 1.
	\end{equation*}
	Consequently, $\strategy_i$ is almost surely constant on $M_i$ w.r.t. $\ObjFun^\QueryObjMean$.
	We apply the reasoning used in the proof of Thm.~\ref{stm:mdp_exp_mean_strategies} to obtain the combined strategy $\strategy'$ which achieves point (i) and switches to $\strategy_i$ upon remaining in $M_i$.

	Now, fix any  $j \in [d]$, $M_i \in \MECs$, and $p, q \in (0, 1)$.
	We have that $\Expectation^{\strategy_i}[\vfreq_a \mid \LtlEventually\LtlAlways M_i] \geq \Expectation^\strategy[\vfreq_a \mid \LtlEventually\LtlAlways M_i]$ by construction.
	Consequently, $\Expectation^{\strategy'}(\ObjFun^\QueryObjMean_j) \geq \Expectation^{\strategy}(\ObjFun^\QueryObjMean_j)$.

	Since $\strategy'$ is MEC-constant, we have $\CVaR_p^{\strategy'}(\ObjFun^\QueryObjMean_j \mid \LtlEventually\LtlAlways M_i) = \Expectation^{\strategy'}[\ObjFun^\QueryObjMean_j \mid \LtlEventually\LtlAlways M_i]$.
	Further, by $\Expectation^\strategy[\vfreq_a \mid \LtlEventually\LtlAlways M_i] \cdot p_i \leq \Expectation^{\strategy_i}[\vfreq_a]$ for all $a$, we get $\Expectation^\strategy[\ObjFun^\QueryObjMean_j \mid \LtlEventually\LtlAlways M_i] \leq \Expectation^{\strategy_i}[\ObjFun^\QueryObjMean_j]$.
	So, $\CVaR_p^{\strategy_i}(\ObjFun^\QueryObjMean_j) = \Expectation^{\strategy_i}[\ObjFun^\QueryObjMean_j] \geq \Expectation^\strategy[\ObjFun^\QueryObjMean_j \mid \LtlEventually\LtlAlways M_i] \geq \CVaR_q^{\strategy}(\ObjFun^\QueryObjMean_j \mid \LtlEventually\LtlAlways M_i)$, as $\CVaR \leq \Expectation$.

	Finally, we apply this inequality together with property (i), obtaining $\CVaR_p^{\strategy}(\ObjFun^\QueryObjMean_j) \leq \CVaR_p^{\strategy'}(\ObjFun^\QueryObjMean_j)$ by Thm.~\ref{stm:cvar_partitioning}.
\end{proof}

We utilize this structural property to design a linear program for these constraints.
However, similarly to the previously considered LPs, it relies on knowing the $\VaR$ for each $\QueryCritCVaR_\QueryProbCVaR$ constraint.
Due to the non-linear behaviour of $\CVaR$, the classical techniques do not allow us to conclude that $\VaR$ is polynomially sized and thus we do not present the \enquote{matching} NP upper bound, but a PSPACE upper bound, which we achieve as follows.

\begin{figure}
	\renewcommand{\labelenumi}{(\arabic{enumi}) }
	\begin{enumerate}
		\item \label{fig:mean_lp_multi:non-neg}
		All variables $y_a$, $y_s$, $x_a$, $x_s$ are non-negative.
		\item \label{fig:mean_lp_multi:transient}
		Transient flow for $s \in \States$:
		\begin{equation*}
			\mathbbm{1}_{\initstate}(s) + {\sum}_{a \in \Actions} y_a \cdot \Trans(a, s) = {\sum}_{a \in \AvAct(s)} y_a + y_s
		\end{equation*}
		\item \label{fig:mean_lp_multi:switching}
		Probability of switching in a MEC is the frequency of using its actions for $M_i \in \MECs$:
		\begin{equation*}
			{\sum}_{s \in M_i} y_s = {\sum}_{a \in M_i} x_a
		\end{equation*}
		\item \label{fig:mean_lp_multi:recurrent}
		Recurrent flow for $s \in \States$:
		\begin{equation*}
			x_s = {\sum}_{a \in \Actions} x_a \cdot \Trans(a, s) = {\sum}_{a \in \AvAct(s)} x_a
		\end{equation*}
		\item \label{fig:mean_lp_multi:satisfaction}
		$\CVaR$ and expectation satisfaction for $j \in [d]$:
		\begin{gather*}
			{\sum}_{s \in \States^j_\leq} x_s \cdot \reward_j(s) + \left(\QueryProbCVaRVec_j  - {\sum}_{s \in \States^j_\leq} x_s\right) \cdot \textbf{t}_j \geq \QueryProbCVaRVec_j \cdot \QueryThreshCVaRVec_j \\
			{\sum}_{s \in \States} x_s \cdot \reward_j(s) \geq \QueryThreshExpVec_j
		\end{gather*}
		\item \label{fig:mean_lp_multi:verify_mec}
		Verify MEC classification guess for $j \in [d]$:
		\begin{align*}
			{\sum}_{s \in M^j_\leq} x_s \cdot \reward_j(s) & \leq \textbf{t}_j \quad \text{ for $M^j_\leq \in \MECs^j_\leq \union \{M^j_=\}$} \\
			{\sum}_{s \in M^j_\geq} x_s \cdot \reward_j(s) & \geq \textbf{t}_j \quad \text{ for $M^j_\geq \in \MECs^j_> \union \{M^j_=\}$}
		\end{align*}
		\item \label{fig:mean_lp_multi:verify_var}
		Verify VaR guess for $j \in [d]$:
		\begin{align*}
			{\sum}_{s \in \States^j_\leq} x_s & \leq \QueryProbCVaRVec_j &
			{\sum}_{s \in \States^j_\leq \union M^j_=} x_s & \geq \QueryProbCVaRVec_j
		\end{align*}
	\end{enumerate}

	\caption{LP used to decide multi-dimensional mean-payoff queries given a guess \textbf{t} of $\VaR_{\QueryProbCVaR_j}$ and MEC classification $\MECs_\leq^j$, $M_=^j$, and $\MECs_>^j$.
		$\States_\sim^j := \{s \in \States \mid s \in M \text{ and } M \in \MECs_\sim^j\}$, $\sim \in \{\leq, >\}$.} \label{fig:mean_lp_multi}
\end{figure}

\begin{theorem}
	$\QueryMDP^{\{\QueryCritExp, \QueryCritCVaR\}}_{\QueryObjMean, \QueryDimMulti}$ is in PSPACE.
\end{theorem}

\begin{proof}[Proof sketch]
	We use the existential theory of the reals, which is NP-hard and in PSPACE \cite{DBLP:conf/stoc/Canny88}, to encode our problem.
	The VaR vector \textbf{t} is existentially quantified and the formula is a polynomially sized program with constraints linear in VaR's and linear in the remaining variables.
	This shows the complexity result.
	
	The details of the procedure are as follows.
	For each $j \in [d]$, we use the existential theory of reals to guess the achieved VaR $\textbf{t} = \VaR_{\QueryProbCVaRVec_j}$.
	Further, we non-deterministically obtain the following polynomially-sized information (or deterministically try out all options in PSPACE).
	For each $j \in [d]$ and for each MEC $M_i$, we guess if the value achieved in $M_i$ is at most (denoted $M_i \in \MECs_\leq^j$) or above (denoted $M_i \in \MECs_>^j$) the respective $\textbf{t}_j$, and exactly one MEC $M^j_=$, which achieves a value equal to it.
	Given these guesses, we check whether the LP in Fig.~\ref{fig:mean_lp_multi} has a solution.
	%Note that the program is linear when $\textbf{t}$ is replaced by a constant.

	Equations~\eqref{fig:mean_lp_multi:non-neg}-\eqref{fig:mean_lp_multi:recurrent} describe the transient flow like the previous LP's and, additionally, the recurrent flow like in \cite[Sec.~9.3]{Puterman-book} or \cite{DBLP:journals/lmcs/EtessamiKVY08,DBLP:journals/corr/abs-1104-3489,DBLP:journals/lmcs/ChatterjeeKK17}.
	This addition is needed, since now our MECs are not trivial, i.e. single state.
	Again, Inequalities~\eqref{fig:mean_lp_multi:satisfaction} verify that the CVaR and expectation constraints are satisfied.
	Finally, Inequalities~\eqref{fig:mean_lp_multi:verify_mec} and \eqref{fig:mean_lp_multi:verify_var} verify the previously guessed information, i.e. the VaR vector and the MEC classification.

	Using the very same techniques, it is easy to prove that solutions to the LP correspond to satisfying strategies and vice versa.
	In particular, Inequalities~\eqref{fig:mean_lp_multi:verify_mec} and \eqref{fig:mean_lp_multi:verify_var} directly make use of the MEC-constant property of Lem.~\ref{stm:mean_multi_mec_constant_strategies}.
\end{proof}

While MEC-constant strategies are sufficient for $\QueryCritExp$ with $\QueryCritCVaR$, in contrast, they are not even for just $\QueryMDP^{\{\QueryCritVaR\}}_{\QueryObjMean, \QueryDimMulti}$ \cite[Ex.22]{DBLP:conf/lics/ChatterjeeKK15}.
Consequently, only an exponentially large LP is known for $\QueryMDP^{\{\QueryCritVaR\}}_{\QueryObjMean, \QueryDimMulti}$.
We can combine all the objective functions together as follows:

\begin{theorem}
	$\QueryMDP^{\{\QueryCritExp, \QueryCritVaR, \QueryCritCVaR\}}_{\QueryObjMean, \QueryDimMulti}$ is in EXPSPACE.
\end{theorem}
\begin{proof}[Proof sketch]
	We proceed exactly as in the previous case, but now the flows in Equality~\eqref{fig:mean_lp_multi:recurrent} are split into exponentially many flows, depending on the set of dimensions where they achieve the given VaR threshold, see LP $L$ in \cite[Fig.~4]{DBLP:conf/lics/ChatterjeeKK15}.
	The resulting size of the program is polynomial in the size of the system and exponential in $d$.
	Hence the call to the decision procedure of the existential theory of reals results in the EXPSPACE upper bound.
\end{proof}

\begin{table*}
	\caption{
	Schematic summary of known and new results.
	Strategies are abbreviated by \enquote{C/$n$-M}, where C is either \emph{D}eterministic or \emph{R}andomizing, $n$ is the size of the memory, and M is either \emph{D}etereministic or \emph{S}tochastic \emph{MEM}ory.
	} %
	\small %
	\begin{tabular}{c|cc||c|cccc}
		$\QueryDim$  &              \multicolumn{2}{c||}{$\QueryDimSingle$}               &                                                            \multicolumn{5}{c}{$\QueryDimMulti$}                                                             \\
		$\QueryObj$  &                 \multicolumn{2}{c||}{\textsf{any}}                 &    $\QueryObjReach$    &                                                \multicolumn{4}{c}{$\QueryObjMean$}                                                 \\
		$\QueryCrit$ & $\cardinality{\QueryCrit} = 1$ & $\cardinality{\QueryCrit} \geq 2$ & $\CVaR \in \QueryCrit$ & $\{\QueryCritExp,\QueryCritVaR_0\}$ & $\{\QueryCritVaR\}$ & $\{\CVaR\}, \{\CVaR, \Expectation\}$ & $\{\Expectation, \CVaR, \VaR\}$ \\ \midrule
		  Complex.   &                      \multicolumn{2}{c||}{P}                       & NP-c., P for fixed $d$ &                  P                  &         EXP         &            NP-h., PSPACE             &         NP-h., EXPSPACE         \\
		   Strat.    &            D/1-MEM             &             R/2-SMEM              &        R/2-SMEM        &                                            \multicolumn{4}{c}{ \rule[0.5\parskip]{2cm}{0.1pt} \quad R/$\infty$-DMEM \quad \rule[0.5\parskip]{2cm}{0.1pt}}
			%\multicolumn{4}{c}{$\underbrace{\hspace{6cm}}_{\text{\small R/$\infty$-DMEM}}$}
	\end{tabular} \label{tbl:results} %
\end{table*}

\section{Conclusion}

We introduced the conditional value-at-risk for Markov decision processes in the setting of classical verification objectives of reachability and mean payoff.
We observed that in the single dimensional case the additional CVaR constraints do not increase the computational complexity of the problems.
As such they provide a useful means for designing risk-averse strategies, at no additional cost.
In the multidimensional case, the problems become NP-hard.
Nevertheless, this may not necessarily hinder the practical usability.
Our results are summarized in Table~\ref{tbl:results}.

We conjecture that the VaR's for given CVaR constraints are polynomially large numbers.
In that case, the provided algorithms would yield NP-completeness for $\QueryMDP^{\{\QueryCritCVaR\}}_{\QueryObjMean, \QueryDimMulti}$ and EXPTIME containment for $\QueryMDP^{\{\QueryCritExp,\QueryCritVaR, \QueryCritCVaR\}}_{\QueryObjMean, \QueryDimMulti}$, where the exponential dependency is only on the dimension. %, not the size of the system. %, where the dependency is polynomial

%% Acknowledgments
\begin{acks}
This research has been partially supported by the Czech Science Foundation grant No.~\mbox{18-11193S} and the German Research Foundation (DFG) project KR 4890/2 ``Statistical Unbounded Verification'' (383882557).
We thank Vojt{\v e}ch Forejt for bringing up the topic of CVaR and the initial discussions with Jan Kr{\v c}\'al and wish them both happy life in industry.
We also thank Michael Luttenberger and the anonymous reviewers for insightful comments and valuable suggestions.
\end{acks}

%% Commands \grantsponsor{<sponsorID>}{<name>}{<url>} and
%% \grantnum[<url>]{<sponsorID>}{<number>} should be used to
%% acknowledge financial support and will be used by metadata
%% extraction tools.
%This material is based upon work supported by the \grantsponsor{GS100000001}{National Science Foundation}{http://dx.doi.org/10.13039/100000001} under Grant No.~\grantnum{GS100000001}{nnnnnnn} and Grant No.~\grantnum{GS100000001}{mmmmmmm}.
%Any opinions, findings, and conclusions or recommendations expressed in this material are those of the author and do not necessarily reflect the views of the National Science Foundation.

%% Bibliography
\bibliography{ref}%


\begin{thebibliography}{35}

%%% the Web address.

\ifx \showCODEN    \undefined \def \showCODEN     #1{\unskip}     \fi
\ifx \showDOI      \undefined \def \showDOI       #1{#1}\fi
\ifx \showISBNx    \undefined \def \showISBNx     #1{\unskip}     \fi
\ifx \showISBNxiii \undefined \def \showISBNxiii  #1{\unskip}     \fi
\ifx \showISSN     \undefined \def \showISSN      #1{\unskip}     \fi
\ifx \showLCCN     \undefined \def \showLCCN      #1{\unskip}     \fi
\ifx \shownote     \undefined \def \shownote      #1{#1}          \fi
\ifx \showarticletitle \undefined \def \showarticletitle #1{#1}   \fi
\ifx \showURL      \undefined \def \showURL       {\relax}        \fi
% invisible to TeX
\providecommand\bibfield[2]{#2}
\providecommand\bibinfo[2]{#2}
\providecommand\natexlab[1]{#1}
\providecommand\showeprint[2][]{arXiv:#2}

\bibitem[\protect\citeauthoryear{Artzner, Delbaen, Eber, and Heath}{Artzner
  et~al\mbox{.}}{1999}]%
        {MAFI:MAFI068}
\bibfield{author}{\bibinfo{person}{Philippe Artzner}, \bibinfo{person}{Freddy
  Delbaen}, \bibinfo{person}{Jean-Marc Eber}, {and} \bibinfo{person}{David
  Heath}.} \bibinfo{year}{1999}\natexlab{}.
\newblock \showarticletitle{Coherent Measures of Risk}.
\newblock \bibinfo{journal}{{\em Mathematical Finance\/}} \bibinfo{volume}{9},
  \bibinfo{number}{3} (\bibinfo{year}{1999}), \bibinfo{pages}{203--228}.
\newblock
\showISSN{1467-9965}


\bibitem[\protect\citeauthoryear{Ashok, Chatterjee, Daca, K{\v
  r}et{\'{\i}}nsk{\'{y}}, and Meggendorfer}{Ashok et~al\mbox{.}}{2017}]%
        {DBLP:conf/cav/AshokCDKM17}
\bibfield{author}{\bibinfo{person}{Pranav Ashok}, \bibinfo{person}{Krishnendu
  Chatterjee}, \bibinfo{person}{Przemyslaw Daca}, \bibinfo{person}{Jan K{\v
  r}et{\'{\i}}nsk{\'{y}}}, {and} \bibinfo{person}{Tobias Meggendorfer}.}
  \bibinfo{year}{2017}\natexlab{}.
\newblock \showarticletitle{Value Iteration for Long-Run Average Reward in
  Markov Decision Processes}. In \bibinfo{booktitle}{{\em {CAV}}} {\em
  (\bibinfo{series}{{LNCS}})}, Vol.~\bibinfo{volume}{10426}.
  \bibinfo{publisher}{Springer}, \bibinfo{pages}{201--221}.
\newblock


\bibitem[\protect\citeauthoryear{Baier, Daum, Dubslaff, Klein, and
  Kl{\"{u}}ppelholz}{Baier et~al\mbox{.}}{2014a}]%
        {DBLP:conf/nfm/BaierDDKK14}
\bibfield{author}{\bibinfo{person}{Christel Baier}, \bibinfo{person}{Marcus
  Daum}, \bibinfo{person}{Clemens Dubslaff}, \bibinfo{person}{Joachim Klein},
  {and} \bibinfo{person}{Sascha Kl{\"{u}}ppelholz}.}
  \bibinfo{year}{2014}\natexlab{a}.
\newblock \showarticletitle{Energy-Utility Quantiles}. In
  \bibinfo{booktitle}{{\em {NFM}}} {\em (\bibinfo{series}{{LNCS}})},
  Vol.~\bibinfo{volume}{8430}. \bibinfo{publisher}{Springer},
  \bibinfo{pages}{285--299}.
\newblock


\bibitem[\protect\citeauthoryear{Baier, Dubslaff, and Kl{\"{u}}ppelholz}{Baier
  et~al\mbox{.}}{2014b}]%
        {DBLP:conf/csl/BaierDK14}
\bibfield{author}{\bibinfo{person}{Christel Baier}, \bibinfo{person}{Clemens
  Dubslaff}, {and} \bibinfo{person}{Sascha Kl{\"{u}}ppelholz}.}
  \bibinfo{year}{2014}\natexlab{b}.
\newblock \showarticletitle{Trade-off analysis meets probabilistic model
  checking}. In \bibinfo{booktitle}{{\em {CSL-LICS}}}.
  \bibinfo{publisher}{{ACM}}, \bibinfo{pages}{1:1--1:10}.
\newblock


\bibitem[\protect\citeauthoryear{Baier, Dubslaff, Kl{\"{u}}ppelholz, Daum,
  Klein, M{\"{a}}rcker, and Wunderlich}{Baier et~al\mbox{.}}{2014c}]%
        {DBLP:conf/fase/BaierDKDKMW14}
\bibfield{author}{\bibinfo{person}{Christel Baier}, \bibinfo{person}{Clemens
  Dubslaff}, \bibinfo{person}{Sascha Kl{\"{u}}ppelholz},
  \bibinfo{person}{Marcus Daum}, \bibinfo{person}{Joachim Klein},
  \bibinfo{person}{Steffen M{\"{a}}rcker}, {and} \bibinfo{person}{Sascha
  Wunderlich}.} \bibinfo{year}{2014}\natexlab{c}.
\newblock \showarticletitle{Probabilistic Model Checking and Non-standard
  Multi-objective Reasoning}. In \bibinfo{booktitle}{{\em {FASE}}} {\em
  (\bibinfo{series}{{LNCS}})}, Vol.~\bibinfo{volume}{8411}.
  \bibinfo{publisher}{Springer}, \bibinfo{pages}{1--16}.
\newblock


\bibitem[\protect\citeauthoryear{Baier, Klein, Kl{\"{u}}ppelholz, and
  Wunderlich}{Baier et~al\mbox{.}}{2017}]%
        {DBLP:conf/tacas/Baier0KW17}
\bibfield{author}{\bibinfo{person}{Christel Baier}, \bibinfo{person}{Joachim
  Klein}, \bibinfo{person}{Sascha Kl{\"{u}}ppelholz}, {and}
  \bibinfo{person}{Sascha Wunderlich}.} \bibinfo{year}{2017}\natexlab{}.
\newblock \showarticletitle{Maximizing the Conditional Expected Reward for
  Reaching the Goal}. In \bibinfo{booktitle}{{\em {TACAS}}} {\em
  (\bibinfo{series}{{LNCS}})}, Vol.~\bibinfo{volume}{10206}.
  \bibinfo{pages}{269--285}.
\newblock


\bibitem[\protect\citeauthoryear{B{\"{a}}uerle and Ott}{B{\"{a}}uerle and
  Ott}{2011}]%
        {DBLP:journals/mmor/BauerleO11}
\bibfield{author}{\bibinfo{person}{Nicole B{\"{a}}uerle} {and}
  \bibinfo{person}{Jonathan Ott}.} \bibinfo{year}{2011}\natexlab{}.
\newblock \showarticletitle{Markov Decision Processes with
  Average-Value-at-Risk criteria}.
\newblock \bibinfo{journal}{{\em Math. Meth. of {OR}\/}} \bibinfo{volume}{74},
  \bibinfo{number}{3} (\bibinfo{year}{2011}), \bibinfo{pages}{361--379}.
\newblock


\bibitem[\protect\citeauthoryear{Beder}{Beder}{1995}]%
        {seductive}
\bibfield{author}{\bibinfo{person}{Tanya~Styblo Beder}.}
  \bibinfo{year}{1995}\natexlab{}.
\newblock \showarticletitle{VAR: Seductive but dangerous}.
\newblock \bibinfo{journal}{{\em Financial Analysts Journal\/}}
  \bibinfo{volume}{51}, \bibinfo{number}{5} (\bibinfo{year}{1995}),
  \bibinfo{pages}{12--24}.
\newblock


\bibitem[\protect\citeauthoryear{Br{\'{a}}zdil, Brozek, Chatterjee, Forejt, and
  Kucera}{Br{\'{a}}zdil et~al\mbox{.}}{2014}]%
        {DBLP:journals/corr/abs-1104-3489}
\bibfield{author}{\bibinfo{person}{Tom{\'{a}}s Br{\'{a}}zdil},
  \bibinfo{person}{V{\'{a}}clav Brozek}, \bibinfo{person}{Krishnendu
  Chatterjee}, \bibinfo{person}{Vojtech Forejt}, {and}
  \bibinfo{person}{Anton{\'{\i}}n Kucera}.} \bibinfo{year}{2014}\natexlab{}.
\newblock \showarticletitle{Two Views on Multiple Mean-Payoff Objectives in
  Markov Decision Processes}.
\newblock \bibinfo{journal}{{\em {LMCS}\/}} \bibinfo{volume}{10},
  \bibinfo{number}{1} (\bibinfo{year}{2014}).
\newblock


\bibitem[\protect\citeauthoryear{Br{\'{a}}zdil, Chatterjee, Forejt, and
  Kucera}{Br{\'{a}}zdil et~al\mbox{.}}{2013}]%
        {DBLP:conf/lics/BrazdilCFK13}
\bibfield{author}{\bibinfo{person}{Tom{\'{a}}s Br{\'{a}}zdil},
  \bibinfo{person}{Krishnendu Chatterjee}, \bibinfo{person}{Vojtech Forejt},
  {and} \bibinfo{person}{Anton{\'{\i}}n Kucera}.}
  \bibinfo{year}{2013}\natexlab{}.
\newblock \showarticletitle{Trading Performance for Stability in Markov
  Decision Processes}. In \bibinfo{booktitle}{{\em {LICS}}}.
  \bibinfo{publisher}{{IEEE} Computer Society}, \bibinfo{pages}{331--340}.
\newblock


\bibitem[\protect\citeauthoryear{Bruy{\`{e}}re, Filiot, Randour, and
  Raskin}{Bruy{\`{e}}re et~al\mbox{.}}{2017}]%
        {DBLP:journals/iandc/BruyereFRR17}
\bibfield{author}{\bibinfo{person}{V{\'{e}}ronique Bruy{\`{e}}re},
  \bibinfo{person}{Emmanuel Filiot}, \bibinfo{person}{Mickael Randour}, {and}
  \bibinfo{person}{Jean{-}Fran{\c{c}}ois Raskin}.}
  \bibinfo{year}{2017}\natexlab{}.
\newblock \showarticletitle{Meet your expectations with guarantees: Beyond
  worst-case synthesis in quantitative games}.
\newblock \bibinfo{journal}{{\em Inf. Comput.\/}}  \bibinfo{volume}{254}
  (\bibinfo{year}{2017}), \bibinfo{pages}{259--295}.
\newblock


\bibitem[\protect\citeauthoryear{Canny}{Canny}{1988}]%
        {DBLP:conf/stoc/Canny88}
\bibfield{author}{\bibinfo{person}{John~F. Canny}.}
  \bibinfo{year}{1988}\natexlab{}.
\newblock \showarticletitle{Some Algebraic and Geometric Computations in
  {PSPACE}}. In \bibinfo{booktitle}{{\em {STOC}}}. \bibinfo{publisher}{{ACM}},
  \bibinfo{pages}{460--467}.
\newblock


\bibitem[\protect\citeauthoryear{Carpin, Chow, and Pavone}{Carpin
  et~al\mbox{.}}{2016}]%
        {DBLP:conf/icra/CarpinCP16}
\bibfield{author}{\bibinfo{person}{Stefano Carpin}, \bibinfo{person}{Yinlam
  Chow}, {and} \bibinfo{person}{Marco Pavone}.}
  \bibinfo{year}{2016}\natexlab{}.
\newblock \showarticletitle{Risk aversion in finite Markov Decision Processes
  using total cost criteria and average value at risk}. In
  \bibinfo{booktitle}{{\em {ICRA}}}. \bibinfo{publisher}{{IEEE}},
  \bibinfo{pages}{335--342}.
\newblock


\bibitem[\protect\citeauthoryear{Chatterjee, Forejt, and Wojtczak}{Chatterjee
  et~al\mbox{.}}{2013}]%
        {DBLP:conf/lpar/ChatterjeeFW13}
\bibfield{author}{\bibinfo{person}{Krishnendu Chatterjee},
  \bibinfo{person}{Vojtech Forejt}, {and} \bibinfo{person}{Dominik Wojtczak}.}
  \bibinfo{year}{2013}\natexlab{}.
\newblock \showarticletitle{Multi-objective Discounted Reward Verification in
  Graphs and MDPs}. In \bibinfo{booktitle}{{\em {LPAR}}} {\em
  (\bibinfo{series}{{LNCS}})}, Vol.~\bibinfo{volume}{8312}.
  \bibinfo{publisher}{Springer}, \bibinfo{pages}{228--242}.
\newblock


\bibitem[\protect\citeauthoryear{Chatterjee, Kom{\'{a}}rkov{\'{a}}, and K{\v
  r}et{\'{\i}}nsk{\'{y}}}{Chatterjee et~al\mbox{.}}{2015}]%
        {DBLP:conf/lics/ChatterjeeKK15}
\bibfield{author}{\bibinfo{person}{Krishnendu Chatterjee},
  \bibinfo{person}{Zuzana Kom{\'{a}}rkov{\'{a}}}, {and} \bibinfo{person}{Jan
  K{\v r}et{\'{\i}}nsk{\'{y}}}.} \bibinfo{year}{2015}\natexlab{}.
\newblock \showarticletitle{Unifying Two Views on Multiple Mean-Payoff
  Objectives in Markov Decision Processes}. In \bibinfo{booktitle}{{\em
  {LICS}}}. \bibinfo{publisher}{{IEEE} Computer Society},
  \bibinfo{pages}{244--256}.
\newblock


\bibitem[\protect\citeauthoryear{Chatterjee, K{\v r}et{\'{\i}}nsk{\'{a}}, and
  K{\v r}et{\'{\i}}nsk{\'{y}}}{Chatterjee et~al\mbox{.}}{2017}]%
        {DBLP:journals/lmcs/ChatterjeeKK17}
\bibfield{author}{\bibinfo{person}{Krishnendu Chatterjee},
  \bibinfo{person}{Zuzana K{\v r}et{\'{\i}}nsk{\'{a}}}, {and}
  \bibinfo{person}{Jan K{\v r}et{\'{\i}}nsk{\'{y}}}.}
  \bibinfo{year}{2017}\natexlab{}.
\newblock \showarticletitle{Unifying Two Views on Multiple Mean-Payoff
  Objectives in Markov Decision Processes}.
\newblock \bibinfo{journal}{{\em {LMCS}\/}} \bibinfo{volume}{13},
  \bibinfo{number}{2} (\bibinfo{year}{2017}).
\newblock


\bibitem[\protect\citeauthoryear{Clemente and Raskin}{Clemente and
  Raskin}{2015}]%
        {DBLP:conf/lics/ClementeR15}
\bibfield{author}{\bibinfo{person}{Lorenzo Clemente} {and}
  \bibinfo{person}{Jean{-}Fran{\c{c}}ois Raskin}.}
  \bibinfo{year}{2015}\natexlab{}.
\newblock \showarticletitle{Multidimensional beyond Worst-Case and Almost-Sure
  Problems for Mean-Payoff Objectives}. In \bibinfo{booktitle}{{\em {LICS}}}.
  \bibinfo{publisher}{{IEEE} Computer Society}, \bibinfo{pages}{257--268}.
\newblock


\bibitem[\protect\citeauthoryear{Courcoubetis and Yannakakis}{Courcoubetis and
  Yannakakis}{1995}]%
        {DBLP:journals/jacm/CourcoubetisY95}
\bibfield{author}{\bibinfo{person}{Costas Courcoubetis} {and}
  \bibinfo{person}{Mihalis Yannakakis}.} \bibinfo{year}{1995}\natexlab{}.
\newblock \showarticletitle{The Complexity of Probabilistic Verification}.
\newblock \bibinfo{journal}{{\em J. {ACM}\/}} \bibinfo{volume}{42},
  \bibinfo{number}{4} (\bibinfo{year}{1995}), \bibinfo{pages}{857--907}.
\newblock


\bibitem[\protect\citeauthoryear{de~Alfaro}{de~Alfaro}{1997}]%
        {dA97a}
\bibfield{author}{\bibinfo{person}{L. de Alfaro}.}
  \bibinfo{year}{1997}\natexlab{}.
\newblock {\em \bibinfo{title}{Formal Verification of Probabilistic Systems}}.
\newblock \bibinfo{thesistype}{Ph.D. Dissertation}. \bibinfo{school}{Stanford
  University}.
\newblock


\bibitem[\protect\citeauthoryear{Etessami, Kwiatkowska, Vardi, and
  Yannakakis}{Etessami et~al\mbox{.}}{2008}]%
        {DBLP:journals/lmcs/EtessamiKVY08}
\bibfield{author}{\bibinfo{person}{Kousha Etessami}, \bibinfo{person}{Marta~Z.
  Kwiatkowska}, \bibinfo{person}{Moshe~Y. Vardi}, {and}
  \bibinfo{person}{Mihalis Yannakakis}.} \bibinfo{year}{2008}\natexlab{}.
\newblock \showarticletitle{Multi-Objective Model Checking of Markov Decision
  Processes}.
\newblock \bibinfo{journal}{{\em {LMCS}\/}} \bibinfo{volume}{4},
  \bibinfo{number}{4} (\bibinfo{year}{2008}).
\newblock


\bibitem[\protect\citeauthoryear{Filar, Krass, and Ross}{Filar
  et~al\mbox{.}}{1995a}]%
        {FKR95}
\bibfield{author}{\bibinfo{person}{J.A. Filar}, \bibinfo{person}{D. Krass},
  {and} \bibinfo{person}{K.W Ross}.} \bibinfo{year}{1995}\natexlab{a}.
\newblock \showarticletitle{Percentile performance criteria for limiting
  average {M}arkov decision processes}.
\newblock \bibinfo{journal}{{\it IEEE Trans. Automat. Control}}
  \bibinfo{volume}{40}, \bibinfo{number}{1} (\bibinfo{date}{Jan}
  \bibinfo{year}{1995}), \bibinfo{pages}{2--10}.
\newblock


\bibitem[\protect\citeauthoryear{Filar, Krass, and Ross}{Filar
  et~al\mbox{.}}{1995b}]%
        {CIS-125853}
\bibfield{author}{\bibinfo{person}{Jerzy~A. Filar}, \bibinfo{person}{Dmitry
  Krass}, {and} \bibinfo{person}{Keith~W. Ross}.}
  \bibinfo{year}{1995}\natexlab{b}.
\newblock \showarticletitle{Percentile performance criteria for limiting
  average Markov decision processes}.
\newblock \bibinfo{journal}{{\it IEEE Trans. Automat. Control}}
  \bibinfo{volume}{40} (\bibinfo{year}{1995}), \bibinfo{pages}{2--10}.
\newblock


\bibitem[\protect\citeauthoryear{Forejt, Kwiatkowska, Norman, Parker, and
  Qu}{Forejt et~al\mbox{.}}{2011}]%
        {DBLP:conf/tacas/ForejtKNPQ11}
\bibfield{author}{\bibinfo{person}{Vojtech Forejt}, \bibinfo{person}{Marta~Z.
  Kwiatkowska}, \bibinfo{person}{Gethin Norman}, \bibinfo{person}{David
  Parker}, {and} \bibinfo{person}{Hongyang Qu}.}
  \bibinfo{year}{2011}\natexlab{}.
\newblock \showarticletitle{Quantitative Multi-objective Verification for
  Probabilistic Systems}. In \bibinfo{booktitle}{{\em {TACAS}}} {\em
  (\bibinfo{series}{{LNCS}})}, Vol.~\bibinfo{volume}{6605}.
  \bibinfo{publisher}{Springer}, \bibinfo{pages}{112--127}.
\newblock


\bibitem[\protect\citeauthoryear{Gilbert, Weng, and Xu}{Gilbert
  et~al\mbox{.}}{2017}]%
        {DBLP:conf/aaai/GilbertWX17}
\bibfield{author}{\bibinfo{person}{Hugo Gilbert}, \bibinfo{person}{Paul Weng},
  {and} \bibinfo{person}{Yan Xu}.} \bibinfo{year}{2017}\natexlab{}.
\newblock \showarticletitle{Optimizing Quantiles in Preference-Based Markov
  Decision Processes}. In \bibinfo{booktitle}{{\em {AAAI}}}.
  \bibinfo{publisher}{{AAAI} Press}, \bibinfo{pages}{3569--3575}.
\newblock


\bibitem[\protect\citeauthoryear{Haase and Kiefer}{Haase and Kiefer}{2015}]%
        {DBLP:conf/icalp/HaaseK15}
\bibfield{author}{\bibinfo{person}{Christoph Haase} {and}
  \bibinfo{person}{Stefan Kiefer}.} \bibinfo{year}{2015}\natexlab{}.
\newblock \showarticletitle{The Odds of Staying on Budget}. In
  \bibinfo{booktitle}{{\em {ICALP}}} {\em (\bibinfo{series}{{LNCS}})},
  Vol.~\bibinfo{volume}{9135}. \bibinfo{publisher}{Springer},
  \bibinfo{pages}{234--246}.
\newblock


\bibitem[\protect\citeauthoryear{Haase, Kiefer, and Lohrey}{Haase
  et~al\mbox{.}}{2017}]%
        {DBLP:conf/lics/HaaseKL17}
\bibfield{author}{\bibinfo{person}{Christoph Haase}, \bibinfo{person}{Stefan
  Kiefer}, {and} \bibinfo{person}{Markus Lohrey}.}
  \bibinfo{year}{2017}\natexlab{}.
\newblock \showarticletitle{Computing quantiles in Markov chains with
  multi-dimensional costs}. In \bibinfo{booktitle}{{\em {LICS}}}.
  \bibinfo{publisher}{{IEEE} Computer Society}, \bibinfo{pages}{1--12}.
\newblock


\bibitem[\protect\citeauthoryear{Huang and Guo}{Huang and Guo}{2016}]%
        {DBLP:journals/siamjo/HuangG16}
\bibfield{author}{\bibinfo{person}{Yonghui Huang} {and}
  \bibinfo{person}{Xianping Guo}.} \bibinfo{year}{2016}\natexlab{}.
\newblock \showarticletitle{Minimum Average Value-at-Risk for Finite Horizon
  Semi-Markov Decision Processes in Continuous Time}.
\newblock \bibinfo{journal}{{\em {SIAM} Journal on Optimization\/}}
  \bibinfo{volume}{26}, \bibinfo{number}{1} (\bibinfo{year}{2016}),
  \bibinfo{pages}{1--28}.
\newblock


\bibitem[\protect\citeauthoryear{Kageyama, Fujii, Kanefuji, and
  Tsubaki}{Kageyama et~al\mbox{.}}{2011}]%
        {DBLP:journals/ajcm/KageyamaFKT11}
\bibfield{author}{\bibinfo{person}{Masayuki Kageyama},
  \bibinfo{person}{Takayuki Fujii}, \bibinfo{person}{Koji Kanefuji}, {and}
  \bibinfo{person}{Hiroe Tsubaki}.} \bibinfo{year}{2011}\natexlab{}.
\newblock \showarticletitle{Conditional Value-at-Risk for Random Immediate
  Reward Variables in Markov Decision Processes}.
\newblock \bibinfo{journal}{{\em American J. Computational Mathematics\/}}
  \bibinfo{volume}{1}, \bibinfo{number}{3} (\bibinfo{year}{2011}),
  \bibinfo{pages}{183--188}.
\newblock


\bibitem[\protect\citeauthoryear{Li, Zhong, and Brandeau}{Li
  et~al\mbox{.}}{2017}]%
        {DBLP:journals/corr/abs-1711-05788}
\bibfield{author}{\bibinfo{person}{Xiaocheng Li}, \bibinfo{person}{Huaiyang
  Zhong}, {and} \bibinfo{person}{Margaret~L. Brandeau}.}
  \bibinfo{year}{2017}\natexlab{}.
\newblock \showarticletitle{Quantile Markov Decision Process}.
\newblock \bibinfo{journal}{{\em CoRR\/}}  \bibinfo{volume}{abs/1711.05788}
  (\bibinfo{year}{2017}).
\newblock


\bibitem[\protect\citeauthoryear{Miller and Yang}{Miller and Yang}{2017}]%
        {DBLP:journals/siamco/MillerY17}
\bibfield{author}{\bibinfo{person}{Christopher~W. Miller} {and}
  \bibinfo{person}{Insoon Yang}.} \bibinfo{year}{2017}\natexlab{}.
\newblock \showarticletitle{Optimal Control of Conditional Value-at-Risk in
  Continuous Time}.
\newblock \bibinfo{journal}{{\em {SIAM} J. Control and Optimization\/}}
  \bibinfo{volume}{55}, \bibinfo{number}{2} (\bibinfo{year}{2017}),
  \bibinfo{pages}{856--884}.
\newblock


\bibitem[\protect\citeauthoryear{Puterman}{Puterman}{1994}]%
        {Puterman-book}
\bibfield{author}{\bibinfo{person}{M.~L. Puterman}.}
  \bibinfo{year}{1994}\natexlab{}.
\newblock \bibinfo{booktitle}{{\em Markov {D}ecision {P}rocesses}}.
\newblock \bibinfo{publisher}{J. Wiley and Sons}.
\newblock


\bibitem[\protect\citeauthoryear{Randour, Raskin, and Sankur}{Randour
  et~al\mbox{.}}{2017}]%
        {DBLP:journals/fmsd/RandourRS17}
\bibfield{author}{\bibinfo{person}{Mickael Randour},
  \bibinfo{person}{Jean{-}Fran{\c{c}}ois Raskin}, {and} \bibinfo{person}{Ocan
  Sankur}.} \bibinfo{year}{2017}\natexlab{}.
\newblock \showarticletitle{Percentile queries in multi-dimensional Markov
  decision processes}.
\newblock \bibinfo{journal}{{\em {FMSD}\/}} \bibinfo{volume}{50},
  \bibinfo{number}{2-3} (\bibinfo{year}{2017}), \bibinfo{pages}{207--248}.
\newblock


\bibitem[\protect\citeauthoryear{Rockafellar and Uryasev}{Rockafellar and
  Uryasev}{2000}]%
        {Rockafellar00optimizationof}
\bibfield{author}{\bibinfo{person}{R.~Tyrrell Rockafellar} {and}
  \bibinfo{person}{Stanislav Uryasev}.} \bibinfo{year}{2000}\natexlab{}.
\newblock \showarticletitle{Optimization of Conditional Value-at-Risk}.
\newblock \bibinfo{journal}{{\em Journal of Risk\/}}  \bibinfo{volume}{2}
  (\bibinfo{year}{2000}), \bibinfo{pages}{21--41}.
\newblock


\bibitem[\protect\citeauthoryear{Rockafellar and Uryasev}{Rockafellar and
  Uryasev}{2002}]%
        {rockafellar2002conditional}
\bibfield{author}{\bibinfo{person}{R~Tyrrell Rockafellar} {and}
  \bibinfo{person}{Stanislav Uryasev}.} \bibinfo{year}{2002}\natexlab{}.
\newblock \showarticletitle{Conditional value-at-risk for general loss
  distributions}.
\newblock \bibinfo{journal}{{\em Journal of banking \& finance\/}}
  \bibinfo{volume}{26}, \bibinfo{number}{7} (\bibinfo{year}{2002}),
  \bibinfo{pages}{1443--1471}.
\newblock


\bibitem[\protect\citeauthoryear{Ummels and Baier}{Ummels and Baier}{2013}]%
        {DBLP:conf/fossacs/UmmelsB13}
\bibfield{author}{\bibinfo{person}{Michael Ummels} {and}
  \bibinfo{person}{Christel Baier}.} \bibinfo{year}{2013}\natexlab{}.
\newblock \showarticletitle{Computing Quantiles in Markov Reward Models}. In
  \bibinfo{booktitle}{{\em FoSSaCS}} {\em (\bibinfo{series}{{LNCS}})},
  Vol.~\bibinfo{volume}{7794}. \bibinfo{publisher}{Springer},
  \bibinfo{pages}{353--368}.
\newblock


\end{thebibliography}
\newpage
%% Appendix
\appendix
\section{Appendix}

\subsection{Properties of the CVaR operator} \label{sec:app:cvar_prop}

\subsubsection{Details of the VaR definition}

In the main body, we mentioned a possible alternative definition of VaR, namely $\VaR_p := \inf \{r \in \Reals \mid F_X(r) \geq p\}$, which is quite symmetric our chosen definition $\VaR_p :=  \sup \{r \in \Reals \mid F_X(r) \leq p\}$.
Indeed, these definitions mostly are equivalent, especially on continuous random variables.
A subtle difference arises for discrete random variables, which we explain in the following.

Recall that in the proof of correctness for the LP in Fig.~\ref{fig:reach_lp}, we mentioned that under particular circumstances, namely $\underline{x}_s = x_s$ for all $s \in \reachSet_\leq$, we have $\VaR_p > t$, $\VaR_p = \min \{r \in \reward(\States) \mid r > t\}$ to be precise.
As we are only dealing with lower bound constraints, this doesn't matter in our case.

Essentially, the problem arises if there is an interval $I$ such that $F(r) = p$ for all $r \in I$.
The $\sup$-based definition of $\VaR_p$ now yields the upper bound of the interval, i.e. $\sup I$, while the $\inf$-based definition yields the lower bound $\inf I$.
This in turn complicates the decision procedure slightly, since we have to exclude the case $\underline{x}_s = 0$ for all $s \in \reachSet_=$ for the guessed VaR bound $t = \QueryThreshVaR$.

\subsubsection{Event-based definition}

Intuitively, one might think of defining $\CVaR$ in terms of the \enquote{worst $p$-quantile events}, i.e., in this setting, the set of worst runs in the system.
Formally, we would be interested in some set $\Omega_p \in \Salgebra$ such that $\measure(\Omega_p) = p$ and $X(\event) \leq \VaR_p(X)$ for all events $\event \in \Omega_p$.
Clearly, this set may not be uniquely defined.
Consider, for example, a system and random variable where all runs attains the same value $v$.
Then, we could choose \emph{any} appropriately sized set of runs as $\Omega_p$.

But, this interpretation has another, more immediate problem.
The amount of distinct runs in a system may easily be finite or even a singleton, for example any single-state system.
Then, no set of runs actually satisfies this condition.
Yet, when interpreting $\CVaR$ as a measure of risk, it may still make sense to include a \enquote{fraction} of some run into this set, or equivalently, give less weight to it.

To fix this problem, one can consider a slight modification of the probability space of the system.
Informally, we attach some uniformly chosen value between $0$ and $1$ to any run of the system.
This ensures that no run has positive weight, yet the probability of any particular sequence of visited states remains unchanged.

\subsubsection{Further insight into the non-linearity} \label{sec:app:cvar:nonlinear}

\begin{example}
	Recall Ex.~\ref{ex:mdp_cvar_hard} and again set $p = 0.2$.
	We claim that any mixed strategy $\strategy_\lambda = \lambda \strategy_a + (1-\lambda) \strategy_b$ for $\lambda \in (0, 1)$ yields a CVaR strictly less than $5$.
	Fix any such lambda.
	Under the mixed strategy $\strategy_\lambda$, a value of $0$ is obtained with probability $(1-\lambda) 0.1$, $5$ with $\lambda$, and $10$ with $(1-\lambda) 0.9$, respectively.
	Together, we have that $\VaR_p = 5$ iff $p < 0.1 (1-\lambda) + \lambda$, i.e. $\lambda > \frac{1}{9}$, and $\VaR_p = 10$ otherwise.
	We handle these two cases separately.
	In the former case, the $\CVaR_p$ is
	\begin{multline*}
		\tfrac{1}{p} \left((1-\lambda) 0.1 \cdot 0 + (p - 0.1 (1-\lambda)) \cdot 5\right) = 5 - 2.5 (1-\lambda) = \\
		= 2.5 (1 + \lambda).
	\end{multline*}
	On the other hand, in the latter we have
	\begin{multline*}
		\tfrac{1}{p} \left((1-\lambda) 0.1 \cdot 0 + \lambda \cdot 5 + (p - ((1-\lambda) 0.1 + \lambda)) \cdot 10 \right) = \\
		= \ldots = 5 (1 - 4 \lambda).
	\end{multline*}
	Together, we get the graph in Fig.~\ref{fig:cvar_nonlinear_graph}.
	
	\begin{figure}[h]
		\begin{tikzpicture}[auto]
		\begin{axis}[axis x line=bottom,
		x axis line style={-},
		x label style={at={(axis description cs:0.5,0)},anchor=north},
		y axis line style={->},
		y label style={at={(axis description cs:0,0)},anchor=south west},
		height=3cm,width={0.5\columnwidth},
		xmin=0,xmax=1,xtick={0,1},xmajorticks=true,xlabel={$\lambda$},
		ymin=0,ymax=6,ytick={0,2.5,5},ylabel={}]
		\draw (axis cs:0,5) -- (axis cs:1/9,2.5) -- (axis cs:1,5);
		\end{axis}
		\end{tikzpicture}
		\caption{Graph of $\CVaR_p(X_\lambda)$ for $X_\lambda$ as in Ex.~\ref{ex:mdp_cvar_hard}.}
		\label{fig:cvar_nonlinear_graph}
	\end{figure}
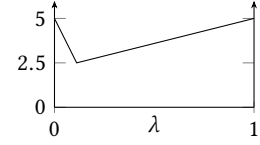
	
\end{example}

\subsubsection{Monotonicity}
%To provide some insight in the $\CVaR$ operator, we prove a general monotonicity property used later on.

\begin{proposition} \label{stm:cvar_monotone}
	Let $p \in (0, 1)$ and $X_1, X_2$ two random variables with CDF $F_1$ and $F_2$, respectively.
	Assume that $F_1$ is point-wise larger than $F_2$ up until $t = \VaR_p(X_1)$, i.e.\ the CDFs satisfy $F_1(r) \geq F_2(r)$ for all $r \in (-\infty, t]$.
	Then we have that $\CVaR_p(X_1) \leq \CVaR_p(X_2)$.
\end{proposition}
\begin{proof}[Proof sketch]
	Let $X_1$ and $X_2$ be random variables and $p \in (0, 1)$ as in the theorem statement.
	Further, assume for now that $X_1$ and $X_2$ are continuous.
	This allows us to choose a measurable, strictly monotone $u : (-\infty, t] \to \Reals$ such that $F_1 \circ u = F_2$.
	Note that necessarily $u(x) \geq x$, since $F_1(x) \geq F_2(x)$ for all $x \leq t$.
	Define $V_i = (-\infty, \VaR_p(X_i)]$ the domain of integration for $\CVaR$.
	Then, by substitution, we get
	\begin{multline*}
		\CVaR_p(X_1) = \int_{V_1} x \ dF_1 = \int_{V_1} u^{-1} \circ u \ dF_1 = \\
		= \int_{u(V_1)} u^{-1} \ dF_1(u) = \int_{V_2} u^{-1} \ dF_2 \leq \int_{V_2} x \ dF_2,
	\end{multline*}
	concluding the proof.

	The general case can be proven analogously, only that one has to slightly relax the conditions on $u$ and deal with the previously mentioned corner cases.
\end{proof}

\begin{corollary}
	Let $p \in (0, 1)$ and $X_1, X_2$ two real valued random variables where $X_2$ stochastically dominates $X_1$.
	Then $\Expectation[X_1] \leq \Expectation[X_2]$, $\VaR_p(X_1) \leq \VaR_p(X_2)$, and $\CVaR_p(X_1) \leq \CVaR_p(X_2)$.
\end{corollary}

\begin{proof}
	For $\Expectation$ and $\VaR$, this immediately follows from the respective definition, for $\CVaR$ the statement follows from Prop.~\ref{stm:cvar_monotone}.
\end{proof}

\subsubsection{CVaR of conditioned variables}

In the proof of Lem.~\ref{stm:mean_multi_mec_constant_strategies}, we applied $\CVaR$ on a random variable together with a conditioning, i.e. $\CVaR_p(X \mid W)$, where $W \in \Salgebra$ is a measurable, non-zero measure event.
Formally, this refers to the CVaR of $X$ on a slightly modified probability space, where $\Prob'$ equals $\Prob$ conditioned on $W$.

For, e.g., expectation, it is known that conditioning on a partitioning of the probability space preserves the total value, i.e. for a finite set $\{ W_i \} \subseteq \Salgebra$ which partitions $\Omega$, we have $\Expectation[X] = \sum \Prob[W_i] \cdot \Expectation[X \mid W_i]$.
This does not hold for CVaR, see Sec.~\ref{sec:app:cvar:nonlinear}.
We instead show a weaker statement, used in the proof of Lem.~\ref{stm:mean_multi_mec_constant_strategies}.

\begin{theorem} \label{stm:cvar_partitioning}
	Fix a probability space $(\Omega, \Salgebra, \Prob)$, a random variable $X$ and $p \in [0, 1]$.
	Further, let $W = \{W_i\} \subseteq \Salgebra$ be a partitioning of $\Omega$, where all $W_i$ are measurable and have non-zero measure.
	Then, there exist $p_i$ and a set $W' \subseteq W$ such that
	\begin{equation*}
		\CVaR_p(X) = \frac{1}{\Prob[\Union_{W_i \in W'} W']}\sum_{W_i \in W'} \Prob[W_i] \CVaR_{p_i}(X \mid W_i)
	\end{equation*}
\end{theorem}

\begin{proof}
	Recall that $\CVaR_1 = \Expectation$ and $\CVaR_0 = \VaR$.
	In these cases, the statement immediately follows:
	For expectation, choose $W' = W$ and $p_i = 1$; for $\VaR$, choose $W' = \{W_i\}$ where $W_i$ is a witness of the worst case and set $p_i = 0$.

	Now, let $p \in (0, 1)$.
	Let $W' = \{W_i \mid \Prob[X \leq v \mid W_i] > 0\}$ be all sets relevant for the CVaR, i.e. which contain \enquote{bad} events.
	Then, we can simply set $p_i = \Prob[X \leq v \mid W_i]$, taking care of potential discrete jumps of $X$ in the respective $W_i$.
\end{proof}

%\begin{example}
%	To demonstrate Thm.~\ref{stm:cvar_partitioning}, consider a random process which first tosses a fair coin and then draws a number uniformly from $[0, 2]$ if head or $[1, 3]$ if tail.
%	It is easy to verify that $\VaR_{0.5} = \frac{3}{2}$ and $\CVaR_{0.5} = \int_0^1 \frac{1}{4} x \, dx + \int_1^{\frac{3}{2}} \frac{1}{2} x \, dx = $.
%	To partition the event space, we choose $W_1 = \text{ \enquote{head}}$ and $W_2 = \text{ \enquote{tail}}$.
%	Observe that 
%\end{example}

\subsection{The general case for weighted reachability} \label{sec:app:cvar:weighted_reach}

\subsubsection{Proof of Thm.~\ref{stm:mdp_exp_reach_strategies_general}}

In the main body, we assumed that the set of target states $\reachSet$ is reached almost surely.
Thm.~\ref{stm:mdp_exp_reach_strategies} shows that in this case memoryless randomizing strategies are sufficient for $\QueryMDP_{\QueryObjReach,\QueryDimSingle}^{\{\QueryCritExp, \QueryCritVaR, \QueryCritCVaR\}}$.
If we drop this assumption, the involved strategies may become more complex.
Intuitively, this is the case since it might pay off to not move to the target set at all, e.g., if $\reward(s) < 0$ for all $s \in \reachSet$.
We show this with a simple variation of the MDP in Fig.~\ref{fig:mean_payoff_memory} and prove that 2-memory strategies are sufficient.
Recall that this is the same type of strategy required for mean payoff, see Thm.~\ref{stm:mdp_exp_mean_strategies}.

\begin{figure}[!h]
	\begin{tikzpicture}[auto,initial text={},node distance=0.5cm]
		\node[initial above,state] (s0) {};
		\node[actionnode,right=of s0] (s0b) {};
		\node[state,above right=0.2cm and 0.75cm of s0b] (s2) {$5$};
		\node[state,below right=0.2cm and 0.75cm of s0b] (s3) {$-5$};

		\path[->]
			(s0) edge[directedge,loop left]    node[action] {$a$} (s0)
			(s0) edge[actionedge]         node[action] {$b$} (s0b)
			(s0b) edge[probedge]          node[prob] {$0.9$} (s2)
			(s0b) edge[probedge,swap]     node[prob] {$0.1$} (s3)
			(s2) edge[directedge,loop right] (s2)
			(s3) edge[directedge,loop right] (s3);
	\end{tikzpicture}
	\caption{Memory is necessary for general weighted reachability}
	\label{fig:reach_general_memory}
\end{figure}
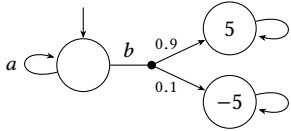 %
\begin{example}
	Consider the MDP presented in Fig.~\ref{fig:reach_general_memory} together with the constraints $\QueryProbVaR = \QueryProbCVaR = 0.05$, $\QueryThreshExp = 1$, and $\QueryThreshCVaR = -3$ (or $\QueryThreshVaR = 0$).
	Using the reasoning of Ex.~\ref{ex:mdp_exp_reach_memory} immediately yields the result, namely that any satisfying strategy needs memory. \QEE
\end{example}

\begin{proof}[Proof of Thm.~\ref{stm:mdp_exp_reach_strategies_general}]
	Since we assumed that all states in $\reachSet$ are absorbing, we can trivially convert this problem to an instance of mean payoff by assigning $\reward(s) = 0$ to all $s \in \States \setminus \reachSet$.
	Applying Thm.~\ref{stm:mdp_exp_mean_strategies} yields the result.
\end{proof}

\subsubsection{Assumptions in Thm.~\ref{stm:mdp_exp_reach_lp_correct}} \label{sec:app:reach_lp_assumptions}

For the LP in Fig.~\ref{fig:reach_lp}, we made several assumptions, namely:
\begin{enumerate}
	\item All MECs consist of a single state and we identify each MEC with the corresponding state, \label{stm:mdp_exp_reach_lp_correct:assumption:mec}
	\item all MECs $m_i$ from which $\reachSet$ is not reachable are considered part of $\reachSet$ and have $\reward(m_i) = 0$, and \label{stm:mdp_exp_reach_lp_correct:assumption:all_reach}
	\item quantile-probabilities are equal, i.e.\ $\QueryProbCVaR = \QueryProbVaR$. \label{stm:mdp_exp_reach_lp_correct:assumption:quantile}
\end{enumerate}

In the following, we present the general procedure to obtain a solution to $\QueryMDP_{\QueryObjReach,\QueryDimSingle}^{\{\QueryCritExp,\QueryCritVaR,\QueryCritCVaR\}}$ queries.
In particular, we
\begin{enumerate}
	\item prove that assumption~\ref{stm:mdp_exp_reach_lp_correct:assumption:mec} and \ref{stm:mdp_exp_reach_lp_correct:assumption:all_reach} can be made w.l.o.g.,
	\item derive a modification of the LP from Fig.~\ref{fig:reach_lp} in Fig.~\ref{fig:reach_lp_ext} which deals with assumption~\ref{stm:mdp_exp_reach_lp_correct:assumption:quantile}, and
	\item show that the combined procedure is correct.
\end{enumerate}

%Recall that we are still are operating under the attraction assumption of Def.~\ref{def:attraction_assumption}.
%For assumption~\ref{stm:mdp_exp_reach_lp_correct:assumption:mec}, observe that it is never profitable to remain in a non-target MEC.

Assumption~\ref{stm:mdp_exp_reach_lp_correct:assumption:all_reach} directly follows from the definition of weighted reachability.
Once a run enters a MEC from which the target set $\reachSet$ is not reachable, the run is guaranteed to achieve a reward of zero.
To resolve the single-state MEC assumption we can lift the problem to the \emph{MEC quotient}~\cite{dA97a}, which satisfies the criterion.
This construction intuitively collapses each MEC into a single representative state and adapts the transitions accordingly.
Since in each MEC every state can be reached from any other with probability $1$, the MEC quotient preserves many infinite horizon properties like (weighted) reachability~\cite{dA97a,Puterman-book,DBLP:conf/cav/AshokCDKM17}.
More precisely, queries can easily be transformed so that they are satisfiable on the MEC quotient if and only if the original query is satisfiable.

\begin{figure}
	\begin{enumerate}
		\renewcommand{\labelenumi}{(\arabic{enumi}) }
		\item \label{fig:reach_lp_ext:non-neg}
		All variables $y_a$, $x_s$, $\underline{x}_s^c$, $\underline{x}_s^v$ are non-negative.
		\item \label{fig:reach_lp_ext:flow}
		Flow for $s \in \States$:
		\begin{equation*}
			\mathbbm{1}_{s_0}(s) + {\sum}_{a \in \Actions} y_a \Trans(a, s) = {\sum}_{a \in \AvAct(s)} y_a + x_s
		\end{equation*}
		\item \label{fig:reach_lp_ext:switching}
		Switching to recurrent behaviour:
		\begin{equation*}
			{\sum}_{s \in M_i} y_s = x_i
		\end{equation*}
		\item \label{fig:reach_lp_ext:var_split}
		$\VaR$-consistent split:
		\begin{align*}
			\underline{x}_s^c & = x_s \text{ for $s \in \reachSet_{<c}$} & \underline{x}_s^c & \leq x_s \text{ for $s \in \reachSet_{=c}$} \\
			\underline{x}_s^v & = x_s \text{ for $s \in \reachSet_{<v}$} & \underline{x}_s^v & \leq x_s \text{ for $s \in \reachSet_{=v}$}
		\end{align*}
		\item \label{fig:reach_lp_ext:consistent_split}
		Probability-consistent split:
		\begin{align*}
			\sum_{s \in \reachSet_{\leq c}} \underline{x}_s^c & = \QueryProbCVaR & \sum_{s \in \reachSet_{\leq v}} \underline{x}_s^v & = \QueryProbVaR
		\end{align*}
		\item \label{fig:reach_lp_ext:satisfaction}
		CVaR and expectation satisfaction:
		\begin{align*}
			\sum_{s \in \reachSet_{\leq c}} \underline{x}_s^c \cdot \reward(s) & \geq \QueryProbCVaR \cdot \QueryThreshCVaR &
			\sum_{s \in \reachSet} x_s \cdot \reward(s) & \geq \QueryThreshExp
		\end{align*}
	\end{enumerate}
	\caption{LP used to decide weighted reachability queries given guesses $t_c$ and $t_v$ of $\VaR_\QueryProbCVaR$ and $\VaR_\QueryProbVaR$, respectively.
	$\reachSet_{\sim c}$ is defined as $\reachSet_{\sim c} := \{s \in \reachSet \mid s \sim t_c\}$ for $\sim~\in \{<, =, \leq\}$, analogous for $\reachSet_{\sim v}$.}
	\label{fig:reach_lp_ext}
	\vspace{-0.5em}
\end{figure}

%First, let us define our decision procedure in Alg.~\ref{alg:reach_single_decide}.
%
%\begin{algorithm}
%	\KwIn{$\MDP = (\States, \Actions, \AvAct, \Trans, \initstate)$ : MDP\newline
%		$\reward : \reachSet \to \Rationals$ : Reward function with target states\newline
%		$\QueryProbCVaR, \QueryProbVaR \in (0, 1)$ : Query probabilities\newline
%		$\QueryThreshExp, \QueryThreshCVaR, \QueryThreshVaR \in \Rationals$ : Query Thresholds}%
%	\KwResult{\textit{yes} iff there exists a strategy satisfying the query}%
%	
%	Compute MEC quotient $\widehat{\MDP}$ (Def.~\ref{def:mec_quotient})\;
%	$\reachSet' \gets \{\hat{s}_i^r\}$ \tcp{Lift $\reachSet$ and $\reward$ to the quotient}
%	\For{$\hat{s}_i^r \in \reachSet'$}{
%		\lIf{$M_i = \{s\}$ and $s \in \reachSet$}{$\reward'(\hat{s}_i^r) \gets \reward(s)$}
%		\lElse{$\reward'(\hat{s}_i^r) \gets 0$}
%	}
%	\For{$t_c \in \reward'(\reachSet')$, $t_v \in \reward'(\reachSet') \setminus (-\infty, \QueryThreshVaR)$}{
%		\lIf{LP of Fig.~\ref{fig:reach_lp_ext} is feasible}{\Return \textit{yes}}
%	}
%	\Return \textit{no}
%	\caption{Algorithm deciding $\QueryMDP_{\QueryObjReach,\QueryDimSingle}^{\{\QueryCritExp,\QueryCritVaR,\QueryCritCVaR\}}$ queries.} \label{alg:reach_single_decide}
%\end{algorithm}

In order to handle assumption~\ref{stm:mdp_exp_reach_lp_correct:assumption:quantile}, we present the adaptation of the LP from Fig.~\ref{fig:reach_lp} in Fig.~\ref{fig:reach_lp_ext}.
Essentially, only an additional type of variable, namely $\underline{x}_s^v$, has been added, verifying the $\VaR_\QueryProbVaR$ constraint.
To prove Thm.~\ref{stm:mdp_exp_reach_lp_correct} for this adapted LP, the exact same reasoning as in the original proof can be applied.

\subsection{Proof of Lem.~\ref{stm:cvar_quasiconvex}} \label{sec:app:quasiconvex}

\begin{proof}
	Let $X_1$, $X_2$, $p$, and $\lambda$ be as in the statement and define $X_\lambda = \lambda X_1 + (1-\lambda) X_2$.
	Further, for $\circ \in \{1, 2, \lambda\}$, set $v_\circ = \VaR_p(X_\circ)$ and $c_\circ = \CVaR_p(X_\circ)$, and let $F_\circ$ denote the CDF of $X_\circ$.
	W.l.o.g., let $v_1 \leq v_2$.
	Observe that $F_\lambda = \lambda F_1 + (1-\lambda) F_2$ and, since CDF are non-decreasing, $F_2(v_1) \leq F_1(v_1) \leq F_1(v_2)$.
	Together, $F_\lambda(v_1) \leq F_\lambda(v_\lambda) \leq F_\lambda(v_2)$, which in turn yields $v_1 \leq v_\lambda \leq v_2$.
	
	Moreover, we can rearrange the definition of $F_\lambda$ to
	\begin{equation*}
		\lambda (F_1(v_\lambda) - p) = (1-\lambda) (p - F_2(v_\lambda)),
	\end{equation*}
	and thus
	\begin{equation*}
		\lambda \int_{v_1}^{v_\lambda} dF_1 = (1-\lambda) \int_{v_\lambda}^{v_2} dF_2.
	\end{equation*}
	Rearranging $c_\lambda$ and splitting the integration domain, we get
	\begin{multline*}
		c_\lambda = \int_{-\infty}^{v_\lambda} v\ dF_\lambda = \lambda \int_{-\infty}^{v_\lambda} v\ dF_1 + (1-\lambda) \int_{-\infty}^{v_\lambda} v\ dF_2 = \\
			= \lambda \left(c_1 + \int_{v_1}^{v_\lambda} v\ dF_1\right) + (1-\lambda) \left(c_2 - \int_{v_\lambda}^{v_2} v\ dF_2\right) \leq \\
			\leq \lambda c_1 + (1-\lambda) c_2 + \left( \lambda v_\lambda \int_{v_1}^{v_\lambda} dF_1 - (1-\lambda) v_\lambda \int_{v_\lambda}^{v_2} dF_2 \right),
	\end{multline*}
	and thus, by the previous equality, $c_\lambda \leq \lambda c_1 + (1-\lambda) c_2$.
\end{proof}

\end{document}